\documentclass[sigplan,nonacm]{acmart}

\usepackage{array}
\usepackage{graphicx}
\usepackage[export]{adjustbox}
\usepackage{amsmath}
\usepackage{tikz}
\usetikzlibrary{quantikz2}
\usetikzlibrary{arrows.meta}

\providecommand{\cb}[1]{\ensuremath{\vcenter{\hbox{#1}}}}

\usepackage{booktabs}

\providecommand{\tbd}{\mbox{--}}

\let\qcoptexttt\texttt
\renewcommand{\texttt}[1]{{\def\_{\textunderscore\allowbreak}\qcoptexttt{#1}}}

\newcommand{\isb}{I\textsuperscript{2}SB}
\AtBeginDocument{%
}
\begin{document}
\title{Bridge of \texorpdfstring{$\Psi$}{Psi}'s: Quantum Circuit Optimization
  with Schr\"odinger Bridges}
\author{Lino S. Hofstetter}
\authornote{Corresponding author}
\email{lhofstette@ethz.ch}
\affiliation{
  \institution{ETH Z\"urich}
  \city{Z\"urich}
  \country{Switzerland}
}
\author{Lia Yeh}
\email{ly404@cam.ac.uk}
\affiliation{
  \institution{University of Cambridge}
  \city{Cambridge}
  \country{United Kingdom}
}
\author{Prakash Murali}
\email{pm830@cam.ac.uk}
\affiliation{
  \institution{University of Cambridge}
  \city{Cambridge}
  \country{United Kingdom}
}
\renewcommand{\shortauthors}{Hofstetter et al.}
\begin{abstract}
Quantum circuit optimization replaces a circuit with an equivalent one of fewer gates and lower depth, reducing execution cost and error rate. 
We ask whether a generative model can learn this transformation directly from examples, rather than selecting from a fixed rewrite library or rigid algebraic routines.
We present Bridge of $\Psi$'s (BOPS), a generative model based on Schr\"odinger bridges, using a custom denoiser architecture, that learns a transformation from a source circuit into an equivalent optimized circuit. 
We train it on data constructed to be hard for existing optimizers, by applying rewrite rules backwards so that each input has a known lower-cost target.
On held-out $8\text{ qubits}\times64 \text{ depth}$ Clifford+$T$ circuits, BOPS reduces gate count by $2.46\times$ and depth by $2.45\times$ in geometric mean, outperforming all nine baseline optimizers.
This constitutes the first generative model bridging quantum circuits and frontier machine learning methods, opening up the quantum compilation stack to learned optimization along multiple axes.
\end{abstract}

\maketitle

\section{Introduction}
\label{sec:introduction}

\begin{figure*}[t]
  \centering
  \includegraphics[max width=\textwidth]{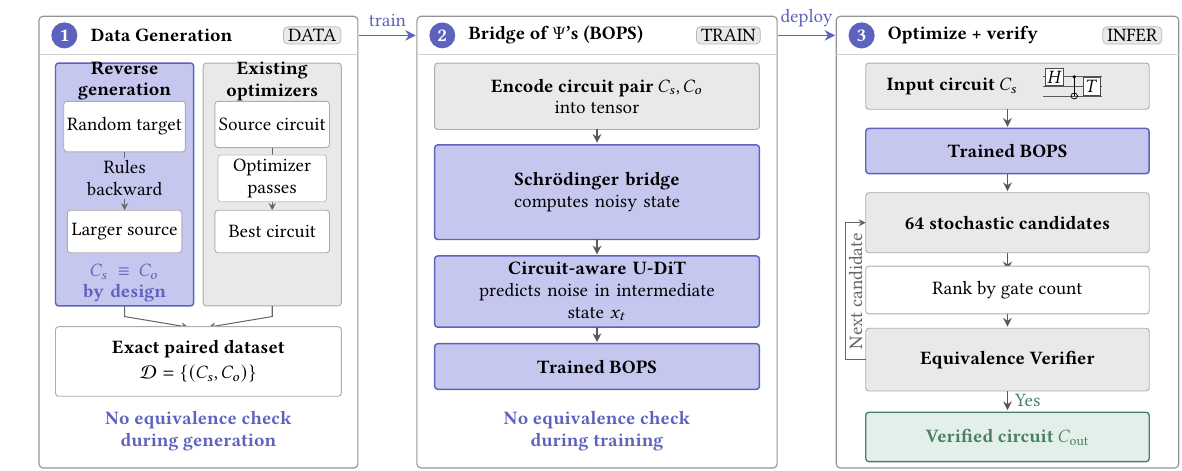}
  \caption{BOPS architecture and training pipeline; new components are shown in lavender.}
  \Description{Block diagram of BOPS. Reverse rewriting produces equivalent source--target circuit pairs; a source-conditioned U-shaped diffusion transformer samples optimized candidates, which are then verified for equivalence. Lavender marks the components introduced by BOPS.}
\label{fig:bops_architecture}
\end{figure*}

Quantum computing promises to extend the range of computations accessible to science, but fulfilling that promise requires new primitives for software--hardware co-design and automated compilation across the quantum systems stack.
Quantum hardware has advanced on several fronts: noisy processors support increasingly large experiments, while recent systems have executed circuits on encoded logical qubits and demonstrated surface-code memories below the error-correction threshold~\cite{preskill2018nisq,bluvstein2024logical,google2025qec}.
Quantum compilation has advanced in parallel, producing end-to-end toolchains and specialized passes for synthesis, optimization, qubit mapping, and routing~\cite{qiskit,sivarajah2020tket,hietala2021voqc,kissinger2020pyzx}.
These optimizations reduce resource requirements and help accelerate progress toward quantum advantage.

Quantum compilers translate hardware-independent quantum programs or circuits into device-executable circuits.
They decompose operations into the device's native gate set, map logical to physical qubits, insert routing operations to satisfy connectivity, and apply semantics-preserving optimizations~\cite{qiskit,sivarajah2020tket}.
However, reliance on hand-crafted analytical rewrites, local heuristics, and restricted search spaces severely limits their ability to generalize and scale across interacting objectives and constraints: gate count, circuit depth, total execution time, connectivity, real-time latency and bottleneck handling, classical--quantum tradeoffs, and heterogeneous fault-tolerance noise and cost models.

Applying recent advances in artificial intelligence to quantum systems could make compilation data-driven, large-scale, and multi-target, but first we must establish that an AI model can learn quantum circuit tasks at all.
AI has already proved effective at the lower level, including quantum-control steering~\cite{Sivak2026learnctrl} and high-throughput error decoding in fault-tolerant architectures~\cite{gu2026neuraldecoders}, and in circuit-level tasks such as Clifford circuit synthesis~\cite{yeung2026clifford} and classical simulation~\cite{koziellpipe2024gnnsim}.

Circuit optimization transforms a circuit into an equivalent lower-cost circuit, where we measure cost by gate count and depth.
Existing methods replace small patterns with cheaper equivalents, resynthesize two- or three-qubit blocks, apply formally verified rewrites, or rewrite graph representations through the ZX calculus~\cite{kissinger2020tcount}.
Qiskit, tket, VOQC, and PyZX combine these techniques~\cite{qiskit,sivarajah2020tket,hietala2021voqc,kissinger2020pyzx}.
Another class searches over sequences of equivalence-preserving transformations, sometimes using synthesized rather than hand-written rewrite rules~\cite{tzap,queso,guoq}.
Every noisy gate introduces error, so the probability of a correct computation generally falls as gate count and depth increase~\cite{preskill2018nisq,pointing2024quanto}.
Under fault tolerance, each logical gate instead becomes a gadget occupying several physical qubits for multiple machine cycles~\cite{litinski2019magic}.
Reducing gate count and depth therefore lowers execution cost in either setting.

Correct optimization must preserve or verify equivalence to the input.
Existing methods avoid repeated checks by restricting reachable circuits to equivalent ones. 
Learned optimizers likewise select predefined transformations or search specialized representations~\cite{foesel2021rl,li2024quarl,ruiz2025alphatensor}.
Moving beyond these restrictions requires checking generated circuits for equivalence---a parallelizable $\text{coNP}$-hard task.

Standard diffusion models generate by transporting Gaussian noise to data. A
Schr\"odinger bridge instead transports between arbitrary data distributions~\cite{debortoli2021dsb,shi2023dsbm}, while paired data lets it learn from corresponding examples~\cite{shi2023dsbm,tong2024sf2m,liu2023i2sb}. We therefore formulate circuit optimization as continuous transport from source circuits to optimized equivalents.
Replacing the noise endpoint with the source circuit compactly specifies the target unitary and supplies a valid implementation from which to start.
Because equivalence is embedded in the data, the bridge emits a complete optimized candidate rather than composing one from predefined rewrites.
This removes unitary construction and equivalence checking from the training loop.
Inference instead performs one fast verification check per candidate and discards inequivalent outputs (\S\ref{sec:verification}).

We implement this formulation as Bridge of $\Psi$'s (BOPS), a standalone optimizer that reads and writes OpenQASM 2 and can serve as a stage in existing compilation pipelines.
Its denoiser is a U-shaped diffusion transformer that resamples only the time axis and keeps every qubit at full resolution. 
One trained network therefore accepts any circuit that fits the grid (\S\ref{sec:denoisers}).
We write $Q\times D$ for that grid, $Q$ qubits by $D$ time steps, so an $8\times64$ circuit acts on eight qubits and has depth at most $64$.
We compare BOPS on gate count and depth with Qiskit~\cite{qiskit}, PyZX~\cite{kissinger2020pyzx}, tket~\cite{sivarajah2020tket}, TZAP~\cite{tzap}, VOQC~\cite{hietala2021voqc}, QUESO~\cite{queso}, and GUOQ~\cite{guoq}, verifying every reported output exactly (\S\ref{sec:verification}, \S\ref{sec:eval}).

\begin{table}[t]
\caption{BOPS and the strongest baseline optimizer on $8\times64$ circuits 
($n=2{,}686$); full results in Table~\ref{tab:q8t64}.}
\label{tab:headline}
\footnotesize
\centering
\setlength{\tabcolsep}{3pt}
\begin{tabular}{@{}lccc@{}}
\toprule
& gates reduced by & depth reduced by & target reached \\
\midrule
QUESO          & $2.11\times$ & $2.11\times$ & $26.1\%$ \\
BOPS & $\mathbf{2.46\times}$ & $\mathbf{2.45\times}$ & $\mathbf{77.2\%}$ \\
\bottomrule
\end{tabular}
\end{table}

\newcommand{\question}[1]{\paragraph{#1}}
For circuit optimization with generative models, we ask:

\question{Can a generative model optimize quantum circuits?}
\emph{Yes.} On generated $8\times64$ Clifford+$T$ circuits, BOPS outperforms every individual baseline in gate count, depth, and target recovery (Table~\ref{tab:headline}).

\question{Does it generalize to unseen circuit structures?}
\emph{Yes.} Although BOPS was not trained on large structured or arithmetic circuits, it remains competitive on circuit benchmarks, showing that its learned optimizations transfer beyond the training distribution (Table~\ref{tab:bench}).

\question{Are off-the-shelf diffusion models effective?}
\emph{No, but they can be adapted.} Standard vision/text diffusion backbones perform poorly (Table~\ref{tab:backbones}). BOPS remedies this by customizing its denoiser architecture (\S\ref{sec:denoisers}) and training regimen (\S\ref{sec:data}).

\question{Does BOPS scale to larger circuits?}
\emph{Yes, by finetuning.} Although circuit space grows superexponentially with the number of qubits (Lemma~\ref{lem:countfixD}), finetuning the $8\times64$ model to $16\times192$ improves $98.4\%$ of larger inputs under windowed inference (Table~\ref{tab:grids}).

\noindent\textbf{Contributions:}
\begin{itemize}
  \item We formulate circuit optimization as a conditional Schr\"odinger bridge from an input circuit to an optimized equivalent, without a fixed rewrite library or equivalence-preserving search space (\S\ref{sec:formulation}).

  \item We build exact training pairs from reverse rewrites and existing optimizers, including reductions those optimizers miss. Encoding equivalence in the data eliminates unitary construction and equivalence checks during training (\S\ref{sec:reverse}, \S\ref{sec:data}).

  \item We adapt U-DiT~\cite{tian2024udit} to variable circuit grids with time-only downsampling, permutation equivariance, and grouped attention, preserving full qubit resolution across widths and depths. The resulting denoiser shortens \textbf{88.9\%} of evaluation circuits, versus 76.2\% for plain U-DiT and 40.6\% for U-Net (Table~\ref{tab:backbones}, \S\ref{sec:denoisers}).

  \item One BOPS model outperforms all nine baselines in gate count and depth reduction, reaching the known target nearly $3\times$ as often as the strongest baseline. Inference uses a fixed number of network evaluations rather than size-dependent rewrite or search trajectories. Every output is exactly verified (Table~\ref{tab:headline}, \S\ref{sec:eval}).
\end{itemize}

\emph{To our knowledge, BOPS is the first generative model to match or outperform state-of-the-art hard-coded and search-based circuit optimizers, opening a path for applying recent AI advances across quantum compilation workflows.}

\section{Background}
\label{sec:background}

\subsection{Quantum circuits and optimization}
\label{sec:bg-circuits}

An $n$-qubit quantum state is represented by a vector with $2^n$ complex entries, and a quantum gate transforms this state through a unitary matrix~\cite{nielsen2010quantum}.
A circuit $C=(g_1,\ldots,g_m)$ is a sequence of gates with combined action represented by a $2^n \times 2^n$ unitary matrix $U(C)$.
We use the Clifford$+T$ gate set:
\begin{equation}
  \mathcal{G} = \{H,\ S,\ S^\dagger,\ T,\ T^\dagger,\ \mathrm{CX}\}.
  \label{eq:gateset}
\end{equation}
This gate set is a standard target for fault-tolerant compilation~\cite{litinski2019magic,ruiz2025alphatensor}.
The gates $H$, $S$, $S^\dagger$, $T$, and $T^\dagger$ act on one qubit, while $\mathrm{CX}$ connects two qubits.

Circuits are commonly drawn as grids in which each row is a qubit and gates are ordered from left to right.
This asymmetry between the qubit and time axes matters to our model's circuit representation.

Two circuits are equivalent, written $C \equiv C'$, when their unitaries differ only by a global phase:
\begin{equation}
  U(C) = e^{i\varphi}U(C').
  \label{eq:equivalence}
\end{equation}
Circuit optimization seeks an equivalent circuit with lower cost:
\begin{equation}
  C^\star = \operatorname*{arg\,min}_{C' \equiv C} \operatorname{cost}(C').
  \label{eq:optproblem}
\end{equation}
Here, we measure cost using gate count and circuit depth.

Although a circuit is stored as a sequence of gates, its unitary contains $4^n$ complex entries.
Explicitly constructing unitaries to compare two circuits therefore scales exponentially with the number of qubits.
Instead, in this work we use two fast but weaker tests for exact circuit equivalence that can be inconclusive (\S\ref{sec:verification}), and in Appendix~\ref{app:equiv} review the scaling of a GPU implementation for conclusive exact equivalence checking---a decision problem that is $\text{coNP}$-hard because 3-SAT can be embedded as reversible Boolean circuits over Clifford+$T$.

\subsection{Diffusion models and Schr\"odinger bridges}
\label{sec:bg-sb}

A diffusion model learns to generate data by reversing a gradual noising process~\cite{ho2020ddpm,song2021sde}.
The forward process starts from a data example $x_0$ and progressively adds Gaussian noise until the state at $t=1$ is approximately random noise.
The reverse process starts from that noise and repeatedly applies a neural network to recover a data sample.
Although the forward process is described as a sequence of steps, the distribution of an intermediate state $x_t$ has a closed form given $x_0$.
Training can therefore choose a random $t$ and sample $x_t$ directly, without generating any earlier states.
The network then predicts the noise separating $x_t$ from $x_0$, and the cost of a training example is independent of the number of denoising steps used during generation~\cite{ho2020ddpm}.

Standard diffusion fixes one endpoint of the process to Gaussian noise.
A Schr\"odinger bridge removes that restriction and instead connects two arbitrary data distributions while remaining as close as possible to a chosen reference diffusion process~\cite{debortoli2021dsb,leonard2014survey}.
Let $\mathcal{P}$ be the reference distribution over paths and $\mathcal{Q}$ a candidate path distribution with endpoint distributions $\pi_0$ and $\pi_1$.
The bridge minimizes how far $\mathcal{Q}$ departs from $\mathcal{P}$:
\begin{equation}
  \mathcal{Q}^\star =
  \operatorname*{arg\,min}_{\mathcal{Q}:\,\mathcal{Q}_0=\pi_0,\,\mathcal{Q}_1=\pi_1}
  D_{\mathrm{KL}}(\mathcal{Q}\,\|\,\mathcal{P}).
  \label{eq:sb}
\end{equation}
In plain terms, the bridge finds a likely random transformation between its endpoint distributions without departing unnecessarily from the reference noising process.
For arbitrary endpoint distributions, this bridge has no general closed form and is usually approximated by repeatedly fitting the process from alternating endpoints~\cite{debortoli2021dsb,shi2023dsbm}.

\isb{} makes training practical by assuming paired endpoints and using a Gaussian reference process with no deterministic drift~\cite{liu2023i2sb}.
Each training example supplies a target $x_0$ together with its corresponding source $x_1$.
Across the dataset, these pairwise bridges define transport between the source and target distributions.

The reference process adds Gaussian noise at a prescribed rate $\beta_t$, called the noise schedule.
The noise schedule defines the variances accumulated before and after time $t$.
\begin{equation}
  \sigma_{\mathrm{fwd}}(t)^2 = \int_0^t \beta_\tau\,\mathrm{d}\tau,
  \qquad
  \sigma_{\mathrm{bwd}}(t)^2 = \int_t^1 \beta_\tau\,\mathrm{d}\tau.
  \label{eq:bridge-variances}
\end{equation}
These variances determine the interpolation weight $\alpha_t$.
\begin{equation}
  \alpha_t =
  \frac{\sigma_{\mathrm{bwd}}(t)^2}
       {\sigma_{\mathrm{fwd}}(t)^2 + \sigma_{\mathrm{bwd}}(t)^2}.
  \label{eq:bridge-alpha}
\end{equation}
Conditioned on the endpoint pair, the distribution of every intermediate state is then known exactly.
\begin{equation}
  \begin{aligned}
    x_t &\sim \mathcal{N}(\mu_t,\sigma_t^2 I), \\
    \mu_t &= \alpha_t x_0 + (1-\alpha_t)x_1, \\
    \sigma_t^2 &= \alpha_t\,\sigma_{\mathrm{fwd}}(t)^2.
  \end{aligned}
  \label{eq:posterior}
\end{equation}
At $t=0$, this distribution collapses to $x_0$, and at $t=1$, it collapses to $x_1$.
Between them, its mean interpolates between the pair while its variance introduces noise.

The key consequence is that training does not need to solve the general bridge problem or simulate a complete path.
For each paired example, training samples a random $t$ and draws $x_t$ directly from the closed-form Gaussian in~\eqref{eq:posterior}.
The denoiser $\epsilon_\theta(x_t,t)$ receives the sampled state and its time.
We train it with the following mean-squared error.
\begin{equation}
  \mathcal{L}(\theta) =
  \mathbb{E}_{t,(x_0,x_1),x_t}\left[
    \left\|
      \epsilon_\theta(x_t,t) -
      \frac{x_t-x_0}{\sigma_{\mathrm{fwd}}(t)}
    \right\|^2
  \right].
  \label{eq:sbloss}
\end{equation}
The target is the displacement from $x_0$ to $x_t$, normalized by the noise accumulated up to time $t$~\cite{liu2023i2sb}.

Generation starts from $x_1$, uses the denoiser to predict $x_0$, and progressively samples earlier states from the corresponding Gaussian posterior~\cite{liu2023i2sb}.
In our application, $x_1$ represents a source circuit and $x_0$ its optimized equivalent.

\section{Design Objectives}
\label{sec:objectives}

\subsection{Related work in circuit optimization}
\label{sec:related}
Classical quantum-circuit optimizers target gate count, gate depth, two-qubit gate count, and related objectives through several broad techniques.
The non-universal $\{\mathrm{CX}, T\}$ fragment of the Clifford+T gate set is known to correspond to Reed--Muller codes.
Techniques focusing on T-count and T-depth optimization---both NP-hard problems~\cite{vanDeWeteringAmy2024}---include TODD~\cite{Heyfron2019todd}, PyZX~\cite{kissinger2020tcount}, FastTODD~\cite{Vandaele2025fasttodd}, and TZAP~\cite{tzap}.

Rule-based pattern-matching tools such as Qiskit~\cite{qiskit}, tket~\cite{sivarajah2020tket}, and VOQC, whose optimizations are verified in the Coq proof assistant~\cite{hietala2021voqc}, use peephole optimizations and templating to rapidly apply local rewrite rules from a fixed library to subcircuits.
Quartz~\cite{quartz}, Quanto~\cite{pointing2024quanto}, and QUESO~\cite{queso} automatically generate and verify rewrite rules for particular gate sets by enumerating small circuits on two to four qubits.
Graph-based rewriting approaches in the ZX-calculus convert circuits to graph representations, which are optimized through rewrite rules in graph form, before converting back to a circuit~\cite{Duncan2020graphthy, kissinger2020pyzx}.
Phase folding (also called phase polynomial, phase gadget, and Pauli exponential) methods are describable in circuit, graph, or symbolic form~\cite{nam2018automated, amy2014tpar, Cowtan2020phasegadget}.
GUOQ~\cite{guoq} combines QUESO's rules with BQSKit's continuous nonlinear resynthesis of blocks of at most three qubits~\cite{BQSKit}.

These approaches largely preserve equivalence by construction, avoiding dense-unitary comparisons that materialize $4^n$ complex entries, but remain limited by their encoded patterns and resynthesis width.
Symbolic checkers avoid this dense representation but can return inconclusive results~\cite{kissinger2020pyzx}.

Automatically generated identities expand the available transformations.
Quanto, for example, finds depth reductions missed by Qiskit and tket~\cite{pointing2024quanto}.
Yet on $n$ qubits there are $2^{\Theta(G\log n)}$ circuit representations (Fig.~\ref{fig:representation}) of $G$-gate circuits and $2^{\Theta(Dn\log n)}$ of circuits with depth at most $D$ (Appendix~\ref{app:prob}).
Practical rule-based optimization is therefore bounded by both the available transformations and the compositions the search can explore.

Learned optimizers can change the search policy without removing these bounds.
Reinforcement learning selects and orders predefined rewrites and can retain cost-increasing moves that greedy search prunes, but its action space remains fixed~\cite{foesel2021rl,li2024quarl}. AlphaTensor-Quantum instead optimizes signature-tensor decompositions, but targets $T$-count and handles structure outside CNOT+$T$ through compilation such as Hadamard gadgetization~\cite{ruiz2025alphatensor, Zen2026realphatensorq}.
Other learned approaches include sequence-to-sequence transformers~\cite{charton2023transformers} and reinforcement learning, using either Proximal Policy Optimization~\cite{Riu2025reinforcement} or Clifford-circuit random walks~\cite{yeung2026clifford}.

Prior diffusion models perform unitary-conditioned synthesis rather than source-conditioned optimization~\cite{furrutter2024diffusion,furrutter2025multimodal}.
AltGraph maps source-circuit DAGs to generated DAGs but reports approximate density-matrix error rather than exact equivalence~\cite{beaudoin2024altgraph}.
Recent work trains an autoregressive transformer for circuit optimization, but reports exact equivalence falling sharply with output length~\cite{saeedi2026autoregressive}.
To our knowledge, BOPS is the first source-conditioned Schr\"odinger bridge for circuit optimization.

\subsection{Design requirements}
\label{sec:towards}
In order to determine how much headroom there is for improvement, we generate optimized-unoptimized circuit pairs by starting with a random quantum circuit, and applying rewrite rules in reverse to expand to a larger equivalence circuit.
For each pair, in Table~\ref{tab:q8t64}, nine baseline optimizers start with the unoptimized circuit, and the strongest of them reaches the reduction of the starting optimized circuit for only $40.7\%$ of pairs.
This gap in headroom motivates a whole-circuit generator that does not select from an explicit inference-time action library.
Composed and nested rewrites can require coordinated changes across distant regions of a circuit.
A paired diffusion model updates the complete representation at each denoising step and learns from exact source--target pairs, requiring neither unitary construction nor equivalence checks during training (\S\ref{sec:bg-sb}).
A circuit-to-circuit bridge can therefore learn optimization from examples while exact inference-time verification enforces correctness.

Source-circuit conditioning makes this generator fit the interface of a compiler pass.
A dense encoding of an $n$-qubit unitary contains $4^n$ complex entries---$65{,}536$ at eight qubits---whereas an $8\times64$ circuit grid contains $512$ token positions (\S\ref{sec:representation}).
Unlike a unitary, the source also supplies an implementation to improve: in their three-qubit unitary-compilation experiments, F\"urrutter et al.\ start from noise and draw $1{,}024$ samples per target~\cite{furrutter2024diffusion}, whereas BOPS starts from a valid source and draws $64$ (\S\ref{sec:setup}).

Circuit grids impose asymmetric representation and denoising requirements.
Gaussian diffusion acts on continuous, fixed-size arrays, whereas circuits are discrete and vary in width and depth. 
Their encoding must preserve multi-qubit connections, support padding, and decode into valid circuits.
Rows identify qubits and columns order gates in time, so downsampling rows can merge wires while purely local operations can miss distant interactions.
The denoiser must therefore keep every qubit at full resolution and combine local patterns with circuit-wide context (\S\ref{sec:denoisers}).

Together, these requirements define our central question: \emph{Can a source-conditioned Schr\"odinger bridge produce equivalent, lower-cost circuits without an explicit inference-time rewrite library, using no unitary construction or equivalence checks during training and a fixed number of network evaluations per fixed-grid input?}

\section{Model and Training}
\label{sec:model}
\subsection{Circuit representation}
\label{sec:representation}
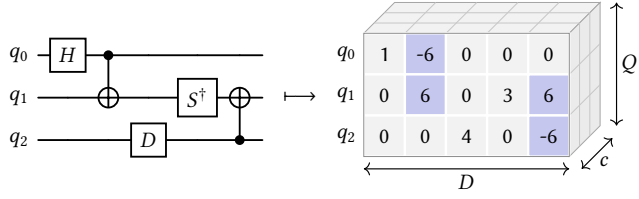
\begin{figure}[!t]
\centering
\definecolor{cellhi}{HTML}{C6C7EF}
\definecolor{cellbg}{HTML}{F2F2F2}
\fontsize{9.5pt}{10.2pt}\selectfont
\adjustbox{max width=\linewidth}{%
\begin{tabular}{@{}c@{\hspace{2mm}}c@{\hspace{2mm}}c@{}}
$\cb{\begin{quantikz}[column sep=0.18cm,row sep={0.64cm,between origins}]
\lstick{$q_0$} & \gate{H} & \ctrl{1} &  &  &  & \\
\lstick{$q_1$} &  & \targ{} &  & \gate{S^\dagger} & \targ{} & \\
\lstick{$q_2$} &  &  & \gate{D} &  & \ctrl{-1} &
\end{quantikz}}$
&
$\cb{\raisebox{0.2ex}{$\longmapsto$}}$
&
$\cb{\begin{tikzpicture}[x=0.62cm,y=-0.62cm,font=\sffamily\fontsize{9.5pt}{10.2pt}\selectfont,
                        every node/.style={inner sep=0pt}]
  \def\dx{0.75}\def\dy{-0.75}
  \fill[cellbg,draw=black!35,line width=0.4pt]
       (0,0) -- (5,0) -- ++(\dx,\dy) -- ++(-5,0) -- cycle;
  \fill[black!8,draw=black!35,line width=0.4pt]
       (5,0) -- (5,3) -- ++(\dx,\dy) -- ++(0,-3) -- cycle;
  \foreach \c in {1,...,4}{\draw[black!22,line width=0.3pt] (\c,0) -- ++(\dx,\dy);}
  \foreach \r in {1,2}{\draw[black!22,line width=0.3pt] (5,\r) -- ++(\dx,\dy);}
  \foreach \s in {0.25,0.5}{%
    \draw[black!22,line width=0.3pt] ({0+\s},{0-\s}) -- ({5+\s},{0-\s});
    \draw[black!22,line width=0.3pt] ({5+\s},{0-\s}) -- ({5+\s},{3-\s});}
  \foreach \r/\c/\v/\f in {%
    0/0/1/cellbg, 0/1/-6/cellhi, 0/2/0/cellbg, 0/3/0/cellbg, 0/4/0/cellbg,
    1/0/0/cellbg, 1/1/6/cellhi,  1/2/0/cellbg, 1/3/3/cellbg, 1/4/6/cellhi,
    2/0/0/cellbg, 2/1/0/cellbg,  2/2/4/cellbg, 2/3/0/cellbg, 2/4/-6/cellhi}{%
      \fill[\f] (\c,\r) rectangle ++(1,1);
      \draw[white,line width=1pt] (\c,\r) rectangle ++(1,1);
      \node at ({\c+0.5},{\r+0.5}) {\v};}
  \draw[black!45,line width=0.4pt] (0,0) rectangle (5,3);
  \draw[<->,line width=0.5pt] (0,3.30) -- node[below,inner ysep=2pt] {$D$} (5,3.30);
  \draw[<->,line width=0.5pt] (5.30,3.30) -- ++(\dx,\dy)
        node[midway,anchor=north west,inner sep=1pt] {$c$};
  \draw[<->,line width=0.5pt] (6.15,{3+\dy}) -- node[right,inner xsep=2pt] {$Q$} (6.15,\dy);
  \foreach \r/\q in {0/0,1/1,2/2}
    \node[anchor=east,inner xsep=3pt] at (0,{\r+0.5}) {$q_\q$};
\end{tikzpicture}}$
\end{tabular}}
\caption{\textbf{Circuit representation}, after F\"urrutter et
al.~\cite{furrutter2024diffusion,furrutter2025multimodal}. Tokens $1$--$6$ are the gates
of~\eqref{eq:gateset}, $0$ an empty cell, and $-6$/$6$ the control/target of a
$\mathrm{CX}$ (shaded).}
\Description{A three-qubit circuit drawn as a quantum circuit diagram on the left maps to a grid of tokens on the right. The grid has one row per qubit and one column per timestep, and is drawn as a three-dimensional block whose third axis is the nine embedding channels. Cells holding the control and target halves of a CX gate are shaded.}
\label{fig:representation}
\end{figure}

Diffusion acts on continuous tensors, whereas a circuit is a discrete sequence of gates, so we adopt the tensor representation of F\"urrutter et al.~\cite{furrutter2024diffusion,furrutter2025multimodal}, shown in Fig.~\ref{fig:representation}.
A circuit becomes a tensor of $Q$ rows and $D$ columns, one row per qubit and one column per timestep, so gates on different qubits share a column, and $Q$ and $D$ are the width and the depth of the circuit.
Every cell holds a token naming the gate that acts on that qubit at that timestep.
A $\mathrm{CX}$ occupies two cells of one column and its control and target are separate tokens, so together with the five one-qubit gates of~\eqref{eq:gateset} and the empty cell, there are eight tokens in total. 
Each token is then replaced, as in~\cite{furrutter2024diffusion}, by a fixed vector of $c = 9$ channels, the vectors mutually orthogonal and each of zero mean and unit variance, giving a tensor of shape $Q \times D \times c$.
Decoding assigns every cell the token whose vector has the highest cosine similarity with it.
A tensor decodes to a valid circuit only if every control has exactly one target in its own column and every target one control.

\subsection{Circuit-to-circuit bridge}
\label{sec:formulation}
\begin{figure*}[!t]
\centering
\definecolor{cellhi}{HTML}{C6C7EF}
\definecolor{srcink}{HTML}{2A6F4E}
\footnotesize
\adjustbox{width=\linewidth}{%
\begin{tikzpicture}[
  x=1cm, y=1cm, font=\sffamily\small,
  flow/.style={-{Stealth[length=4pt]}, line width=0.5pt, black!70},
  skip/.style={-{Stealth[length=4pt]}, line width=0.5pt, black!55, densely dashed},
]
\def\bh{1.1}

\colorlet{topface}{cellhi!55!white}
\colorlet{sideface}{cellhi!78!black}
\foreach \x/\w/\tl/\d in {%
  0.00/2.20/{$D$}/0.12,
  2.66/1.55/{$D/2$}/0.24,
  5.66/0.60/{$D/16$}/0.36,
  7.82/1.55/{$D/2$}/0.24,
  9.95/2.20/{$D$}/0.12}{%
  \fill[topface,draw=black!45,line width=0.4pt]
       (\x,\bh) -- ++(\w,0) -- ++(\d,\d) -- ++(-\w,0) -- cycle;
  \fill[sideface,draw=black!45,line width=0.4pt]
       ({\x+\w},0) -- ++(0,\bh) -- ++(\d,\d) -- ++(0,-\bh) -- cycle;
  \fill[cellhi] (\x,0) rectangle ++(\w,\bh);
  \foreach \r in {1,...,7}{\draw[white,line width=0.5pt] (\x,{\r*\bh/8}) -- ++(\w,0);}
  \draw[black!45,line width=0.4pt] (\x,0) rectangle ++(\w,\bh);
  \node[anchor=south,inner sep=1.5pt,align=center,font=\sffamily\small]
        at ({\x+\w/2+\d},{\bh+\d+0.04}) {\tl};
}
\foreach \x in {4.91, 7.14} \node[font=\sffamily\small,black!55] at (\x,{\bh/2}) {$\cdots$};

\draw[{Stealth[length=3pt]}-{Stealth[length=3pt]},line width=0.4pt,black!65]
      (0.22,0.07) -- (0.22,{\bh-0.07});
\node[anchor=west,inner xsep=2pt,black!65,font=\sffamily\small] at (0.28,{\bh/2}) {$Q$ qubits};
\draw[{Stealth[length=3pt]}-{Stealth[length=3pt]},line width=0.4pt,black!65]
      (0.02,-0.14) -- (2.18,-0.14);
\node[anchor=north,inner ysep=2pt,black!65,font=\sffamily\small] at (1.10,-0.16) {$D$ columns};

\draw[{Stealth[length=3pt]}-{Stealth[length=3pt]},line width=0.4pt,black!65]
      (6.32,-0.14) -- ++(0.36,0.36);
\node[inner sep=1pt,black!65,font=\sffamily\small] at (6.63,-0.09) {$c$};

\foreach \a/\b in {2.32/2.66, 4.45/4.71, 5.14/5.66, 6.62/6.94, 7.37/7.82, 9.61/9.95}
  \draw[flow] (\a,{\bh/2}) -- (\b,{\bh/2});

\foreach \a/\b/\y/\h in {2.07/10.24/1.22/1.55, 4.26/8.22/1.34/1.00}
  \draw[skip] (\a,\y) .. controls (\a,{\y+\h}) and (\b,{\y+\h}) .. (\b,\y);
\node[anchor=south,inner sep=1pt,black!55,font=\sffamily\normalsize] at (6.15,2.48)
      {skip connections};

\node[anchor=east,inner xsep=2pt,font=\sffamily\small,align=right] at (-1.55,{\bh/2+0.36})
      {$x_t$: partly optimized\\grid, $Q\!\times\!D\!\times\!c$};
\node[anchor=east,inner xsep=2pt,srcink,font=\sffamily\small,align=right] at (-1.55,{\bh/2-0.40})
      {$x_1$: source\\circuit};
\draw[flow] (-1.50,{\bh/2+0.36}) -- (-1.05,{\bh/2+0.36}) -- (-1.05,{\bh/2});
\draw[flow,srcink] (-1.50,{\bh/2-0.40}) -- (-1.05,{\bh/2-0.40}) -- (-1.05,{\bh/2});
\draw[flow] (-1.05,{\bh/2}) -- (0,{\bh/2});
\node[anchor=south,inner sep=1.5pt,black!70,font=\sffamily\small] at (-0.52,{\bh/2+0.04}) {concat};
\draw[flow] (12.27,{\bh/2}) -- (12.74,{\bh/2});
\node[anchor=west,inner xsep=2pt,font=\sffamily\small,align=left] at (12.70,{\bh/2})
      {$\epsilon_\theta(x_t,t,x_1)$:\\predicted noise\\grid, $Q\!\times\!D\!\times\!c$};
\end{tikzpicture}}
\caption{\textbf{The denoiser}, drawn as the tensor shape $Q\times D\times c$ at each stage. The tensor of \S\ref{sec:representation} enters as $x_t$, with the source $x_1$ concatenated along the channel axis.
Four downsamplings merge neighboring columns and widen the channels ($192, 256, 384, 512, 768$ over five levels), four mirror levels undo them, and the qubit axis is never reduced. 
The output is the predicted noise $\epsilon_\theta(x_t,t,x_1)$, from which \eqref{eq:x0hat} gives $\hat{x}_0$.}
\Description{}
\label{fig:backbone}
\end{figure*}
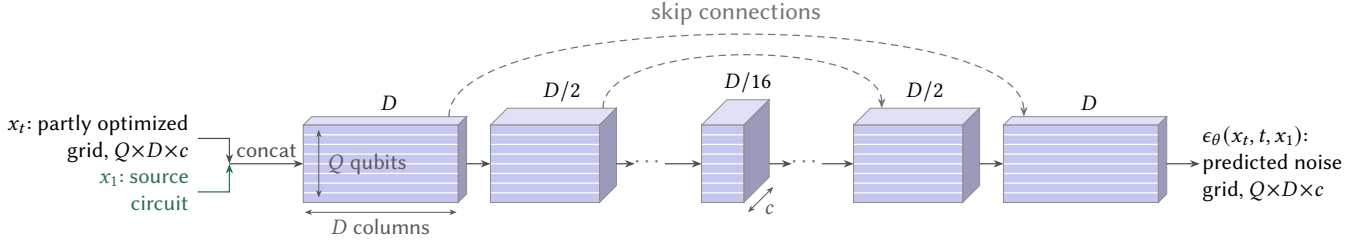

We instantiate the paired bridge of \S\ref{sec:bg-sb} on these tensors, with $x_0$ the tensor of the optimized circuit and $x_1$ that of its source.
The only change is the conditioning: the denoiser has to know which source it is optimizing, so $\epsilon_\theta(x_t,t)$ becomes $\epsilon_\theta(x_t,t,x_1)$, with $x_1$ concatenated to $x_t$ along the channel axis, which gives $2c$ channels and puts the source at every position of the tensor.

Training is otherwise that of~\eqref{eq:sbloss}, with $t$ discretized into $256$ steps and $\beta_t$ from the symmetric quadratic schedule of~\citet{liu2023i2sb}, which is largest at the midpoint, where the state is furthest from either endpoint.
Generation runs from $t=1$, where the state is the source, down to $t=0$.
Rearranging the training target of~\eqref{eq:sbloss} turns the denoiser output into an estimate of the optimized circuit,
\begin{equation}
  \hat{x}_0 = x_t - \sigma_{\mathrm{fwd}}(t)\,\epsilon_\theta(x_t,t,x_1),
  \label{eq:x0hat}
\end{equation}
and the state at the next step $s < t$ is drawn from the bridge posterior between $x_t$ and that estimate,
\begin{equation}
  \begin{aligned}
    x_s &\sim \mathcal{N}\!\left(
      \lambda\hat{x}_0 + (1-\lambda)x_t,\;
      \lambda\,\sigma_{\mathrm{fwd}}(s)^2 I\right), \\
    \lambda &= 1 - \frac{\sigma_{\mathrm{fwd}}(s)^2}{\sigma_{\mathrm{fwd}}(t)^2},
  \end{aligned}
  \label{eq:sample}
\end{equation}
where $\lambda$ is the fraction of the noise accumulated by $t$ that this step
removes.
The chain stays continuous throughout, with nothing projecting $\hat{x}_0$ onto the token vectors between steps.
Consecutive steps need not be adjacent, so it is subsampled to $128$ of the $256$ steps~\cite{zheng2024dbim} at $\eta=1$, where that sampler is exactly~\eqref{eq:sample}.
The loop runs over the schedule and not over gates or cells, so a candidate costs those $128$ denoiser evaluations whatever the circuit holds, with only the cost of one evaluation growing with its length.
Each run is stochastic, so $64$ candidates are drawn per source circuit and verified (\S\ref{sec:verification}).
Beyond the model's attention span, the same procedure runs on blocks and the merged circuit is verified (\S\ref{sec:scaling}).

\subsection{Denoiser architecture}
\label{sec:denoisers}
We build on U-DiT, a U-shaped diffusion transformer for images~\cite{tian2024udit}.
Its encoder levels halve both spatial axes and widen the channels, a mirrored decoder undoes them, and skip connections join matching levels.
Each level runs a stack of blocks comprising attention and then a feed-forward network, both conditioned on the denoising step through adaLN-Zero~\cite{peebles2023dit}.
Attention is downsampled as well.
The token map is split into four $2\times$-downsampled phases.
Each attends within itself, and the results are merged at a quarter of the cost of full attention.
Both the levels and the attention therefore downsample the qubit axis as well as time, which \S\ref{sec:objectives} rules out for circuits.
Table~\ref{tab:backbones} reports the unmodified U-DiT and a standard diffusion U-Net trained on the same pairs.

\paragraph{Time-only downsampling.}
The column index is a coordinate in time, but the row index is only a label.
Qubit order is arbitrary and a $\mathrm{CX}$ acts on an arbitrary pair of rows, so pooling adjacent rows imposes a locality the circuit does not have (\S\ref{sec:towards}).
We downsample the time axis only (Fig.~\ref{fig:backbone}).
Within each row, adjacent columns are concatenated in pairs, $[B,Q,D,C_\ell] \to [B,Q,D/2,2C_\ell]$. 
A linear layer applied at every position projects $2C_\ell$ to the next level's width $C_{\ell+1}$.
These widths are $192, 256, 384, 512, 768$.
The decoder undoes the projection and the concatenation, then merges the skip through a learned per-channel gate and projects back to $C_\ell$.
Neither step ever combines two rows, so a token at level $\ell$ is $2^{\ell}$ consecutive timesteps of a single qubit, and all $Q$ rows reach level $4$, the bottleneck.
Pooling rows would also fix $Q$ into the architecture, whereas downsampling time alone leaves it a free dimension, which is what makes the $16$-qubit models of \S\ref{sec:scaling} a finetune rather than a retrain.
Attention encodes position with rotary embeddings (RoPE)~\cite{su2024rope}, indexed by absolute circuit timestep rather than by position within the level.
The token at index $j$ of level $\ell$ takes position $j\cdot2^{\ell}$ and not $j$, so a cropped grid is the prefix of a full one.

\paragraph{Permutation equivariance.}
Nothing described so far reads the row index, so the network is equivariant under permutations of the qubit axis. 
Two rows carrying the same content receive the same prediction, so where the output must distinguish them, as a $\mathrm{CX}$ must pick one of a pair as its control, the equivariant answer averages the two and decodes to neither.
We therefore add qubit identity from a learned table of $64$ rows and width $C_\ell$, one per level.
The level-$0$ table enters at the input.
The deeper tables enter at every stage that runs at their level and are zero-initialized, so the network begins carrying qubit identity from the input alone.
Equivariance is then learned rather than built in, from the qubit permutations appended to every training pair (\S\ref{sec:curriculum}).
The table is sized for $64$ qubits. 
Training has used eight, or sixteen after the finetunes of \S\ref{sec:scaling}.

\paragraph{Grouped attention.}
Level $\ell$ holds $D_\ell = D/2^{\ell}$ columns, cut into $f = D_\ell/g$ groups of $g$ each, with $f$ rounded down to a divisor of $D_\ell$ and $f=1$ where the level is shorter than $g$.
A token attends to the $gQ$ tokens of its own group and to none outside it, so with $g$ fixed the cost of attention grows linearly with $D$ rather than quadratically.
A token at level $\ell$ covers $2^{\ell}$ timesteps, so one group reaches $g\cdot2^{\ell}$ timesteps of the circuit, twice as many at each level down.
We set $g$ per network, but $f$ is derived per call, so the same weights attend globally on a short circuit and in bands on a long one.
At $D=64$, $g=64$ makes every level a single group.
At $D=512$, $g=32$ reaches $32$ timesteps at level $0$ and the whole circuit at level $4$.
Full attention at $D=512$ would instead cost $2.7\times$ as much per evaluation.
Consecutive blocks alternate a contiguous window~\cite{liu2021swin} with U-DiT's interleaved split~\cite{tian2024udit}.
Here the split uses every $f$-th column rather than a fixed four-way split of both axes.
Similarity is cosine with a learned per-head scale~\cite{liu2022swinv2}.
Every block keeps U-DiT's depthwise convolutions before attention and inside the feed-forward network, here along time only.
Their width is $9$ at level $0$, $7$ at level $1$, and $5$ below, with one more before each downsampling and after each upsampling.

\paragraph{Size and cost.}
The encoder and the decoder each run six blocks at level $0$ and four at levels $1$ to $3$, with four more at the bottleneck, giving $40$ blocks and $98.8$M parameters.
The $D=64$ and the $D=512$ networks share this count and differ only in $g$.
Parameters concentrate at the bottleneck: the twelve blocks at level $0$ hold $6.4\%$ of them against $33.5\%$ for the four at the bottleneck, sixteen times fewer per block.
Compute is spread far more evenly: each transition halves the token count while roughly doubling $C^2$, so everything but attention costs about the same at both ends of the U.
Attention is the exception, four times larger at level $0$, where a block costs $3.2$ GFLOP against $2.6$ at the bottleneck.
The twelve fine-level blocks are therefore the most expensive stage, taking $34\%$ of the $113$ GFLOP for a full $8\times512$ grid against $9\%$ for the bottleneck.
The extra depth we run at the finest level therefore costs time rather than parameters.
Because the U has four transitions, $D$ must be divisible by $16$.
Appendix~\ref{app:reprod} gives the evaluation commands and the baseline settings, and the code release the remaining flags.

\subsection{Training}
\label{sec:training}

All models use the same training recipe.
We use AdamW with $(\beta_1,\beta_2)=(0.9,0.95)$, no weight decay, and gradients clipped at norm $1$.
The learning rate peaks at $4\times10^{-4}$ after $2{,}000$ warmup steps, then cosine-decays to one tenth.
The global batch is $512$ pairs, and the loss of~\eqref{eq:sbloss} is averaged over the columns a source circuit occupies and the next $16$ columns, so the remaining padding is never trained on.
We set $\beta_{\max}=12$, so the schedule of \S\ref{sec:formulation} accumulates a total variance of $1.20$ and $\sigma_t$ peaks at $0.55$ at the midpoint of the bridge, against token vectors of unit variance.
Circuits are batched by length, each batch cropped to one of a fixed set of lengths, two at $D=64$ and eight from $32$ to $512$ columns at $D=512$.
Cropping removes $6.2\times$ of the padded work at $D=512$ and needs no change to the network, which derives $f$ from the length it is handed and positions tokens by absolute column (\S\ref{sec:denoisers}).
We validate and sample from an exponential moving average of the weights with decay $0.9995$.
We report the last checkpoint rather than selecting one, since validation loss selects badly here.
Training runs to $150{,}000$ steps.

\subsection{Equivalence checking}
\label{sec:verification}

No equivalence check enters the training loop, since equivalence is established when each training pair is generated (\S\ref{sec:reverse}, \S\ref{sec:curriculum}).
At inference we decode the $64$ candidates drawn for a source circuit, rank them by gate count with invalid decodings last, and accept the first that passes verification.
A circuit with no accepted candidate is reported at its source length, so nothing we report rests on an unverified circuit.

Two circuits are equivalent when they implement the same unitary up to a global phase.
With $W = U(C)^\dagger U(C')$ and $\tau = \mathrm{Tr}(W)/2^{Q}$, that is exactly the condition $|\tau| = 1$, and $1 - |\tau|^{2}$ is the process infidelity.
Up to $Q = 10$ we build both matrices and read $\tau$ off directly, which settles every pair as equivalent or different.

At $Q = 16$ that matrix has $2^{32}$ entries, so we never build it and check the \texttt{q16} datasets with a \emph{tiered} pipeline instead: a sequence of tests on the ZX diagram of $W$, cheapest first, each of which either answers or declines because the stabilizer decomposition it would need exceeds a fixed budget.
The first test to answer decides the pair, and a pair no test answers is undecided.
The cheapest rewrites the diagram with PyZX~\cite{kissinger2020pyzx} and asks whether it reduces to the identity, which proves equivalence when it succeeds and says nothing when it fails.
The second closes the diagram by wiring each output onto its own input, so that its scalar is $\mathrm{Tr}(W)$ and $\tau$ follows exactly.
A pair is undecided as well when a circuit uses a gate the pipeline cannot translate, or when it exceeds a wall-clock limit of $120$ seconds.
Undecided candidates are not accepted, so the rates of \S\ref{sec:scaling} are lower bounds (Table~\ref{tab:grids}).
In Appendix~\ref{app:equiv}, we summarize another approach available to BOPS for definitively answering equivalence checks through classical simulation with stabilizer decompositions~\cite{Qassim2021alpha,kissinger2022stabdecomp}, which, due to a very fast GPU-parallelized implementation~\cite{sutcliffe2024paramzx}, we would expect to comprise an inconsequential fraction of the total generative inference time.

\section{Training Data}
\label{sec:data}
\begin{figure*}[t]
\centering
\includegraphics[width=\linewidth]{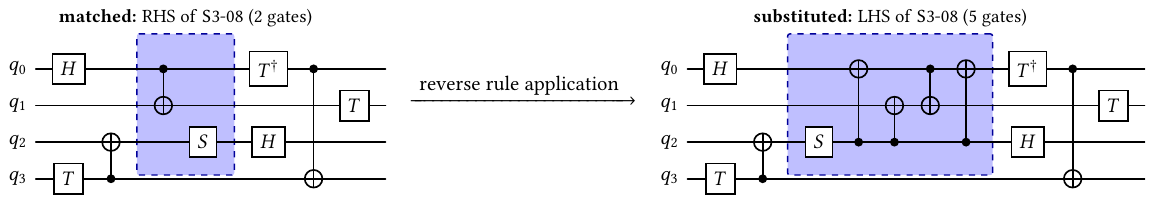}
\caption{\textbf{Reverse rule application.} The SAT-mined rule \textsf{S3-08} (Fig.~\ref{fig:sat-3q}) read right to left: its boxed two-gate right-hand side, matched on $q_0,q_1,q_2$ of a four-qubit circuit, is replaced by its five-gate left-hand side.
}
\Description{Two four-qubit circuit diagrams side by side. In the left one a dashed box marks two gates on qubits zero to two, the right-hand side of rule S3-08. In the right one that box has been replaced by a wider dashed box holding five gates, the rule read backwards, leaving the rest of the circuit unchanged.}
\label{fig:reverse-rewrite}
\end{figure*}

BOPS trains on pairs of a source circuit and an optimized equivalent, so data generation must scale to millions of pairs. 
Building them with existing optimizers has two limitations.
Their targets contain only the reductions those optimizers already recover, so anything beyond would have to come from generalizing across their outputs, and search-based tools take seconds per circuit. 
Finding a minimum-size or minimum-depth equivalent over Clifford$+T$ is NP-hard~\cite{vanDeWeteringAmy2024}, while a heuristic that does not preserve equivalence by construction must verify every accepted candidate.
We therefore construct exact pairs largely by construction.

\subsection{Reverse generation}
\label{sec:reverse}

We generate most pairs in reverse from rewrite rules on one to three qubits.
A rule relates two gate sequences implementing the same operator. 
Applying it from longer to shorter reduces a circuit, and the opposite direction is a \emph{reverse rewrite} (Fig.~\ref{fig:reverse-rewrite}).
Generation has three steps (Fig.~\ref{fig:bops_architecture}): draw a small random circuit, reduce it under the rules, and repeatedly apply reverse rewrites, some nested inside earlier ones.
The reduced circuit becomes the target and the expanded circuit the source. 
Because every rule is verified in advance, they are equivalent by construction and generation never builds a circuit unitary.
Circuit unoptimization~\cite{mori2025unoptimization} expands circuits in the same spirit, inserting two-qubit unitaries and their inverses, but to build compiler benchmarks rather than labeled pairs.

Target reduction and source expansion use different rule sets.
The \emph{expansion} set contains $86$ rules: $36$ handwritten and $50$ solver-mined~\cite{peham2023satclifford}; $83$ shorten circuits, three only rearrange them, and $16$ of the reducing rules have an empty right-hand side.
The \emph{reduction} set combines the same $36$ handwritten rules with $2{,}727$ mined rules, for $2{,}763$ total.
Both sides of every rule are compared as dense unitaries at the rule's minimum width up to global phase, so an invalid rule fails before generating data.

The smaller expansion set produces stronger sources.
Reversing one of its rules adds $3.9$ gates on average, against $1.35$ for a mined rule. 
Weaker expansions produce sources that existing optimizers already solve.
During expansion, matched rewrites are preferred over empty-right-hand-side insertions, which create standalone identity gadgets.
Each gate may then move up to eight positions past gates on disjoint qubits, preserving the operator while separating the rewritten gates in the circuit grid (Appendix~\ref{app:rules}).
These choices embed and disperse the added structure without changing the represented operator.

The target generator determines the distribution BOPS learns.
It combines uniformly random gates, random motifs, and repeated copies of one motif to represent recurring algorithmic structure (Fig.~\ref{fig:motifs}).
Motifs are small subcircuits such as Toffoli decompositions, parity ladders, and phase gadgets. Repetition applies one motif to several groups of qubits in sequence.
The resulting targets contain both unstructured gates and patterns repeated in compiled programs.

\subsection{Dataset composition}
\label{sec:curriculum}

Seven procedures contribute the following shares of the \texttt{q8t64} training split:

\begin{itemize}
\item \emph{Atomic} applies one reverse rewrite ($5.9\%$).
\item \emph{Few}, \emph{medium}, and \emph{hard} apply $2$--$3$, $4$--$8$, and $9$--$20$ reverse rewrites ($17.9\%$, $20.4\%$, and $19.3\%$). Each rewrite is aimed inside an earlier span with probabilities $0.25$, $0.5$, and $0.8$, respectively, increasing both rewrite count and nesting.
\item \emph{Chain} concatenates several independently generated source--target pairs ($22.9\%$). Its sources average $54.2$ gates against $41.9$ overall, increasing length without increasing the number of rewrites within each constituent pair.
\item \emph{Natural} samples circuits randomly from the target generator and uses them as sources rather than targets, labeling each with the shortest output of the nine used optimizers ($9.0\%$).
\item \emph{Clifford} omits $T$ gates, allowing each target to be reconstructed from the source's stabilizer tableau rather than reduction rules ($4.6\%$).
\end{itemize}
Together, they vary rewrite count and nesting, circuit length, source distribution, and gate class.

\subsection{Resulting datasets}
\label{sec:datasets}

The main corpus, \texttt{q8t64}, uses a $Q=8$ by $D=64$ grid with circuits spanning three to eight qubits.
It contains $1{,}953{,}433$ training pairs and $3{,}656$ pairs in each validation and test split.
Sources average $41.9$ gates over $24.9$ occupied columns, against $18.0$ gates for targets. 
The $64$-column source limit caps expansion at $20$ reverse rewrites.

The targets are shorter than the seven-configuration ensemble's outputs on $56.7\%$ of test circuits and longer on $7.7\%$.
Including QUESO and GUOQ changes these rates to $33.9\%$ and $23.1\%$.
These rates motivate reporting both \emph{target reached} and \emph{target beaten}: targets expose unrecovered reductions but are not assumed optimal.

No exact source or target appears in multiple procedure outputs; source deduplication retains one target per source.
Qubit relabeling is the only near-duplicate not excluded: $9$ of the $7{,}312$ evaluation sources relabel a training source.
An additional sample check constructed both unitaries and found every sampled pair equivalent.

Two larger datasets support the scaling experiments of \S\ref{sec:scaling}.
\texttt{q8t512} applies the same seven procedures on a $Q=8$ by $D=512$ grid, accommodating $200$ reverse rewrites; its test sources average $164$ gates over $105$ columns.
\texttt{q16t192} uses $Q=16$, $D=192$, and $1{,}097{,}747$ training pairs but omits \emph{natural}; it is the only corpus whose targets the ensemble never shortens.

\begin{table*}[t]
\caption{
\texttt{q8t64} test split, $n=2{,}686$.
\emph{target reached}/\emph{beaten}: fraction of outputs matching or improving on the source's corpus target.
\emph{cost}: seconds per circuit, one CPU core for the optimizers, one GH200 for BOPS.
}
\label{tab:q8t64}
\footnotesize
\centering
\begin{tabular}{@{}lrrrrrrr@{}}
\toprule
system & gates reduced by & depth reduced by & improved & gap closed & target reached & target beaten & cost (s) \\
\midrule
  Qiskit~O3 & $1.22\times$ & $1.17\times$ & 88.4\% & 32.0\% & 4.4\% & 0.1\% & 0.03 \\
  PyZX-b & $1.76\times$ & $1.67\times$ & 94.2\% & 65.8\% & 12.6\% & 1.9\% & 0.03 \\
  PyZX-fr & $1.99\times$ & $1.97\times$ & 81.4\% & 63.3\% & 21.5\% & 1.4\% & 0.03 \\
  PyZX-fo & $1.86\times$ & $1.86\times$ & 73.8\% & 54.7\% & 17.9\% & 0.5\% & 0.04 \\
  tket & $1.09\times$ & $1.07\times$ & 64.8\% & 13.3\% & 0.2\% & 0.0\% & 0.03 \\
  VOQC & $1.27\times$ & $1.21\times$ & 91.4\% & 36.6\% & 5.5\% & 0.1\% & 0.03 \\
  T$|$zap$\rangle$ & $1.28\times$ & $1.21\times$ & 92.6\% & 38.1\% & 4.3\% & 0.4\% & 0.04 \\
  QUESO & $2.11\times$ & $2.11\times$ & 94.5\% & 86.0\% & 26.1\% & \textbf{14.6\%} & 5.09 \\
  GUOQ & $1.28\times$ & $1.26\times$ & 41.7\% & 33.2\% & 9.9\% & 4.5\% & 4.17 \\
\midrule
  corpus target & $2.66\times$ & $2.72\times$ & 100.0\% & 100.0\% & 100.0\% & 0.0\% & \tbd \\
\midrule
  BOPS, whole & $2.28\times$ & $2.26\times$ & 88.9\% & 88.8\% & \textbf{78.6\%} & 4.8\% & 4.8 \\
  BOPS, windowed & $\boldsymbol{2.46\times}$ & $\boldsymbol{2.45\times}$ & \textbf{96.5\%} & \textbf{95.2\%} & 77.2\% & 11.7\% & 14.4 \\
\bottomrule
\end{tabular}
\end{table*}

\section{Experimental Setup}
\label{sec:setup}

\paragraph{Hardware and software.}
Training and inference run on Isambard-AI~\cite{isambard2025}, on eight NVIDIA GH200 GPUs across two nodes under distributed data parallelism, in bfloat16 with \texttt{torch.compile} on PyTorch $2.11$, CUDA $12.8$ and Python $3.12$, peaking at $10.4$~GiB per GPU for the $D=64$ network and $59.4$ for the $D=512$ one.
Dataset generation runs on $72$-core CPU nodes and every baseline on one core, one circuit at a time, so a per-circuit cost is a single-core cost.

\paragraph{Model and training.}
We evaluate the network of \S\ref{sec:denoisers} at $98.8$ million parameters, trained with the recipe of \S\ref{sec:training} for $150{,}000$ steps at $135$\,ms per step, or $45$ GPU-hours.
The $8\times512$ network is the same architecture for $300{,}000$ steps at $153$\,ms ($102$ GPU-hours), and the $16$-qubit models are $50{,}000$-step finetunes at $63$ and $94$\,ms on $16\times64$ and $16\times192$ ($7$ and $10$ GPU-hours).

\paragraph{Data and benchmarks.}
We use the \texttt{q8t64} splits of \S\ref{sec:datasets}, a $Q=8$ by $D=64$ grid over the gate pool of~\eqref{eq:gateset}, scoring the four rule-generated components \emph{atomic}, \emph{few}, \emph{medium} and \emph{hard}, $2{,}686$ of the $3{,}656$ test sources, unless stated otherwise.
Beyond the test split we evaluate on $228$ vendored flat-QASM circuits that prior optimizer papers report, counted once each where the suites overlap: $69$ Feynman~\cite{amy2014tpar,nam2018automated}, $124$ GUOQ~\cite{guoq}, $31$ QASMBench, three FTCircuitBench and one QUESO~\cite{queso}, with $34$ of the Feynman entries \texttt{feynopt -O3} reductions of the other $35$.
The benchmark tables score the $68$ that fit the $8\times64$ grid and the $119$ that fit $8\times512$.

\paragraph{Baselines.}
We compare against nine configurations of seven optimizers, run on the same circuits and verified exactly as ours are: Qiskit \texttt{transpile} at level three; PyZX \texttt{basic\_optimization}, \texttt{full\_reduce}, and \texttt{full\_optimize}, the last including TODD; tket \texttt{Remove\allowbreak Redundancies}; the VOQC Clifford$+T$ passes; T$|$zap$\rangle$; QUESO; and GUOQ. GUOQ adds BQSKit resynthesis to QUESO.
QUESO and GUOQ run from the GUOQ distribution, QUESO with resynthesis disabled, under a $5$~s search budget and without GUOQ's default $100$-gate cap; T$|$zap$\rangle$ runs under a $60$~s per-circuit timeout and the remaining six to completion.
Stronger tket passes are excluded because they answer through general TK1 rotations, in a different gate set.
Seven of the nine answer on every test circuit; Qiskit returns a non-equivalent circuit on $0.2\%$ and GUOQ on $0.5\%$, both rejected by verification, and GUOQ returns nothing on another $21.6\%$ after its resynthesis server became unreachable, scored as unchanged, so its row is a lower bound. The exact settings for all nine configurations are in Appendix~\ref{app:extlib}.

\paragraph{Inference.}
Whole-circuit mode draws $64$ candidates at $128$ denoising steps each, ranks them by gate count, and verifies them in that order until finding an exactly equivalent candidate (\S\ref{sec:verification}).
That is $8{,}192$ denoiser evaluations per circuit whatever its size, so the sampling budget is fixed.
Windowed mode applies the same procedure to $32$-column blocks and keeps the shorter of the whole-circuit draw and the merged windows (\S\ref{sec:scaling}).

\paragraph{Metrics.}
A system's \emph{reduction}, reported as \emph{gates} and \emph{depth reduced by}, is the geometric mean over circuits of source cost over output cost, so $2.00\times$ means a typical output is half as long and every circuit counts once whatever its size.
The \emph{improvement rate}, reported as \emph{improved}, is the fraction of circuits whose output is strictly shorter than the source and exactly equivalent to it.
Against each dataset target we report the \emph{gap closed}, $(\mathrm{src}-\mathrm{out})/(\mathrm{src}-\mathrm{target})$, with \emph{target reached} and \emph{target beaten}, the fractions matching or improving on the target length, which is not a ceiling, so gap closed can exceed one (\S\ref{sec:datasets}).
The per-circuit benchmark tables give \emph{src}, the source gate count, and \emph{cut}, the gates the best of the nine removes.
A circuit with no verified improvement is reported at its source length, for BOPS and every baseline alike.
Two models trained on the same data with different seeds differ by $1.6$ percentage points in improvement rate, which sets the resolution of every comparison below.

\section{Results and Discussion}
\label{sec:eval}

\begin{table*}[t]
\caption{Each dataset on the rule-generated part of its own test split.
\emph{undecided}: fraction of sources that ended with no verified candidate and at least one the checker could not settle, so their failure may be the checker's.
The last row is the \texttt{q16t64} model scored on the $8$-qubit test set, not a further training run.}
\label{tab:grids}
\footnotesize
\centering
\begin{tabular}{@{}llrrrrrrr@{}}
\toprule
& & & \multicolumn{3}{c}{whole} & \multicolumn{3}{c}{windowed} \\
\cmidrule(lr){4-6}\cmidrule(lr){7-9}
dataset & init & $n$ & undecided & improved & gates reduced by & undecided & improved & gates reduced by \\
\midrule
\texttt{q8t64} ($8\times64$)      & scratch  & 2{,}686 & 0.0\%  & 88.9\% & $2.28\times$ & 0.0\%  & 96.5\% & ${2.46\times}$ \\
\texttt{q16t64} ($16\times64$)    & finetune & 3{,}105 & 7.7\%  & 89.3\% & $1.84\times$ & 3.3\%  & 96.2\% & $1.97\times$ \\
\texttt{q16t192} ($16\times192$)  & finetune & 3{,}136 & 23.9\% & 72.0\% & $1.78\times$ & 1.6\%  & 98.4\% & $ {2.24\times}$ \\
\texttt{q8t512} ($8\times512$)    & scratch  & 4{,}938 & 0.0\%  & 45.4\% & $1.61\times$ & 0.0\%  & 88.7\% & $1.98\times$ \\
\midrule
\texttt{q16t64} on \texttt{q8t64} &          & 2{,}686 & 0.0\%  & 91.2\% & $2.06\times$ & 0.0\%  & {98.8\%} & $2.29\times$ \\
\bottomrule
\end{tabular}
\end{table*}

\begin{table}[t]
\caption{
Denoiser architectures trained on \texttt{q8t64} for $150{,}000$ steps.
Parameter counts cannot be matched exactly, as each architecture constrains its own width.
\emph{improved}: percentage of the evaluated circuits where the model returns a shorter output circuit (\S\ref{sec:quality}).
}
\label{tab:backbones}
\footnotesize
\centering
\begin{tabular}{@{}lrrr@{}}
\toprule
& parameters & steps/s & improved \\
\midrule
BOPS U-DiT    &  98.8\,M &  7.40 & \textbf{88.9\%} \\
UNet           &  90.0\,M & 20.67 & 40.6\% \\
plain U-DiT~\cite{tian2024udit}   & 102.9\,M &  4.71 & 76.2\% \\
\bottomrule
\end{tabular}
\end{table}

\begin{table}[tp]
\caption{Benchmark circuits shortened by BOPS, plus those shortened by none of nine baselines. Bold: BOPS matches or beats the best baseline. \emph{src}/\emph{cut}: source/removed gates; \emph{in}: QASMBench (B), Feynman (F), GUOQ (G), or QUESO (Q). Appendix Table~\ref{tab:bench-full} adds depth and the best optimizer. BOPS uses $128$ candidates over five windowed rounds.}
\label{tab:bench}
\fontsize{8.2pt}{8.4pt}\selectfont
\newcommand{\benchcode}[1]{{\fontsize{8.1pt}{8.1pt}\selectfont #1}}
\setlength{\tabcolsep}{3pt}
\centering
\begin{tabular}{@{}lcrrrr@{}}
\toprule
& & & & \multicolumn{2}{c}{gates} \\
\cmidrule(lr){5-6}
circuit & in & src & cut & best & BOPS \\
\midrule
  lpn\_n5 & \benchcode{B} & 11 & 8 & $3.67\times$ & $\boldsymbol{3.67\times}$ \\
  bb84\_n8 & \benchcode{B} & 54 & 34 & $2.70\times$ & $\boldsymbol{3.00\times}$ \\
  qec\_en\_n5 & \benchcode{B} & 25 & 15 & $2.50\times$ & $2.27\times$ \\
  hs4\_n4 & \benchcode{B} & 40 & 24 & $2.50\times$ & $2.00\times$ \\
  iswap\_n2 & \benchcode{B} & 12 & 4 & $1.50\times$ & $\boldsymbol{1.50\times}$ \\
  rd32-v1\_68 & \benchcode{G} & 42 & 21 & $2.00\times$ & $1.45\times$ \\
  deutsch\_n2 & \benchcode{B} & 8 & 2 & $1.33\times$ & $\boldsymbol{1.33\times}$ \\
  rd32-v0\_66 & \benchcode{G} & 34 & 15 & $1.79\times$ & $1.31\times$ \\
  mod5\_4 & \benchcode{FGQ} & 66 & 29 & $1.78\times$ & $1.25\times$ \\
  ham3\_102 & \benchcode{GQ} & 20 & 4 & $1.25\times$ & $1.18\times$ \\
  barenco\_tof\_3 & \benchcode{FGQ} & 60 & 16 & $1.36\times$ & $1.15\times$ \\
  4mod5-v1\_22 & \benchcode{GQ} & 24 & 4 & $1.20\times$ & $1.14\times$ \\
  tof\_4 & \benchcode{FGQ} & 75 & 15 & $1.25\times$ & $1.12\times$ \\
  qiskit-alu-v0\_26 & \benchcode{G} & 87 & 17 & $1.24\times$ & $1.12\times$ \\
  decod24-v0\_38 & \benchcode{G} & 54 & 10 & $1.23\times$ & $1.10\times$ \\
  alu-bdd\_288 & \benchcode{G} & 87 & 26 & $1.43\times$ & $1.10\times$ \\
  alu-v0\_26 & \benchcode{G} & 87 & 17 & $1.24\times$ & $1.10\times$ \\
  barenco\_tof\_4 & \benchcode{FGQ} & 114 & 28 & $1.33\times$ & $1.10\times$ \\
  4gt13\_91 & \benchcode{G} & 103 & 25 & $1.32\times$ & $1.10\times$ \\
  4gt13-v1\_93 & \benchcode{G} & 74 & 12 & $1.19\times$ & $1.09\times$ \\
  4mod5-v0\_18 & \benchcode{G} & 75 & 32 & $1.74\times$ & $1.09\times$ \\
  decod24-bdd\_294 & \benchcode{G} & 88 & 37 & $1.73\times$ & $1.09\times$ \\
  alu-v0\_27 & \benchcode{GQ} & 39 & 4 & $1.11\times$ & $1.08\times$ \\
  4mod5-bdd\_287 & \benchcode{G} & 79 & 18 & $1.30\times$ & $1.08\times$ \\
  rd32\_270 & \benchcode{G} & 93 & 46 & $1.98\times$ & $1.08\times$ \\
  alu-v3\_35 & \benchcode{GQ} & 40 & 4 & $1.11\times$ & $1.08\times$ \\
  alu-v4\_37 & \benchcode{GQ} & 40 & 4 & $1.11\times$ & $1.08\times$ \\
  4gt11\_82 & \benchcode{Q} & 27 & 3 & $1.12\times$ & $1.08\times$ \\
  mod5d2\_64 & \benchcode{G} & 56 & 11 & $1.24\times$ & $1.08\times$ \\
  alu-v1\_29 & \benchcode{GQ} & 43 & 2 & $1.05\times$ & $\boldsymbol{1.07\times}$ \\
  alu-v2\_33 & \benchcode{GQ} & 43 & 2 & $1.05\times$ & $\boldsymbol{1.07\times}$ \\
  4mod5-v1\_23 & \benchcode{G} & 72 & 38 & $2.12\times$ & $1.07\times$ \\
  4gt5\_76 & \benchcode{G} & 91 & 19 & $1.26\times$ & $1.07\times$ \\
  simon\_n6 & \benchcode{B} & 62 & 48 & $4.43\times$ & $1.07\times$ \\
  4gt13\_92 & \benchcode{GQ} & 66 & 13 & $1.25\times$ & $1.06\times$ \\
  alu-v3\_34 & \benchcode{GQ} & 55 & 6 & $1.12\times$ & $1.06\times$ \\
  mod5mils\_65 & \benchcode{GQ} & 38 & 9 & $1.31\times$ & $1.06\times$ \\
  4mod5-v0\_19 & \benchcode{GQ} & 38 & 5 & $1.15\times$ & $1.06\times$ \\
  mod5\_4\_\_o3 & \benchcode{F} & 40 & 12 & $1.43\times$ & $1.05\times$ \\
  alu-v1\_28 & \benchcode{GQ} & 40 & 4 & $1.11\times$ & $1.05\times$ \\
  decod24-v2\_43 & \benchcode{G} & 61 & 5 & $1.09\times$ & $1.05\times$ \\
  tof\_3 & \benchcode{FGQ} & 45 & 10 & $1.29\times$ & $1.05\times$ \\
  toffoli\_n3 & \benchcode{B} & 24 & 1 & $1.04\times$ & $\boldsymbol{1.04\times}$ \\
  rd53\_138 & \benchcode{G} & 132 & 23 & $1.21\times$ & $1.04\times$ \\
  decod24-v1\_41 & \benchcode{G} & 91 & 14 & $1.18\times$ & $1.03\times$ \\
  barenco\_tof\_4\_\_o3 & \benchcode{F} & 68 & 6 & $1.10\times$ & $1.03\times$ \\
  cat\_state\_n4 & \benchcode{B} & 4 & 0 & $1.00\times$ & $\boldsymbol{1.00\times}$ \\
  ex-1\_166 & \benchcode{GQ} & 22 & 0 & $1.00\times$ & $\boldsymbol{1.00\times}$ \\
  ex1\_226 & \benchcode{G} & 13 & 0 & $1.00\times$ & $\boldsymbol{1.00\times}$ \\
  graycode6\_47 & \benchcode{G} & 5 & 0 & $1.00\times$ & $\boldsymbol{1.00\times}$ \\
  qrng\_n4 & \benchcode{B} & 4 & 0 & $1.00\times$ & $\boldsymbol{1.00\times}$ \\
  teleportation\_n3 & \benchcode{B} & 8 & 0 & $1.00\times$ & $\boldsymbol{1.00\times}$ \\
  tof\_3\_\_o3 & \benchcode{F} & 35 & 0 & $1.00\times$ & $\boldsymbol{1.00\times}$ \\
  tof\_4\_\_o3 & \benchcode{F} & 55 & 0 & $1.00\times$ & $\boldsymbol{1.00\times}$ \\
\bottomrule
\end{tabular}
\end{table}

Tables~\ref{tab:q8t64}, \ref{tab:grids}, \ref{tab:backbones}, and~\ref{tab:bench} report the measurements, with the per-optimizer and per-length breakdowns in Appendix~\ref{app:tables}.

\subsection{Optimization quality}
\label{sec:quality}

BOPS returns shorter circuits than every optimizer we compare against, and its margin over the strongest of them is larger than the margin between that optimizer and the rest of the field (Table~\ref{tab:q8t64}).
One structural difference is the scale of the transformations available to each method. 
QUESO enumerates synthesized rules only up to three qubits and four gates~\cite{queso}, while GUOQ limits both rewriting and resynthesis to three qubits~\cite{guoq}. 
BOPS instead learns from source--target pairs separated by rewrites that compose and nest within one another (\S\ref{sec:reverse}).
Our model matches the target's gate count on $77.2\%$ of sources (target reached) against QUESO's $26.1\%$ and finds an optimization shorter than the target for $11.7\%$ of sources, second to QUESO's $14.6\%$ and far above the $0.0$--$4.5\%$ of the other eight.
The circuits BOPS returns remove on average $95.2\%$ of the gates their corpus target proves can be removed (gap closed). 
Recovering that much of a reduction on circuits BOPS has never seen shows our model learns nearly everything its training data has to teach.

The result depends not only on the training data and the bridge formulation but also on the denoiser that has to learn the transport.
Table~\ref{tab:backbones} isolates that choice by training three denoisers on the same dataset, with the same recipe and the same step count.
Previous diffusion work on circuit synthesis denoises with a U-Net~\cite{furrutter2024diffusion}, and at a comparable parameter count a standard diffusion U-Net~\cite{ronneberger2015unet} improves only $40.6\%$ of our evaluation circuits.
The published U-DiT, the state of the art in U-shaped diffusion transformers~\cite{tian2024udit}, improves $76.2\%$.
Our denoiser, a U-DiT adapted to circuits (\S\ref{sec:denoisers}), improves $88.9\%$, more than double the U-Net and $12.7$ points above U-DiT at $1.6\times$ its step rate.

\subsection{Generalization}
\label{sec:scaling}

Benchmarks have structure not seen during training, as opposed to the test split evaluated in \S\ref{sec:quality}, which makes them a good measure for generalization.
Of the $68$ benchmark circuits that fit the $8\times64$ grid, BOPS shortens $46$, matching the best of the nine baselines on depth at $1.08\times$ and sitting inside their $1.08$--$1.26\times$ band on gates at $1.13\times$ (Tables~\ref{tab:bench} and~\ref{tab:app-bench-per-opt}).
At $8\times512$ it shortens $87$ of the $119$ that fit, at $1.09\times$ on gates and $1.04\times$ on depth against bands of $1.04$--$1.22\times$ and $1.00$--$1.08\times$ (Tables~\ref{tab:app-bench-per-opt} and~\ref{tab:bench512}).
That BOPS matches existing optimizers in depth reduction on $8\times64$ circuits and otherwise performs within their range demonstrates generalization beyond its training distribution, while suggesting that state-of-the-art performance may require training on source circuits whose structure more closely resembles the benchmarks.

A useful optimizer should also adapt efficiently beyond its original training grid.
A $50{,}000$-step finetune of the $8\times64$ model costs $7$ GPU-hours, compared with $45$ for the original training, and achieves $89.3\%$ whole-circuit improvement at $16\times64$ (Table~\ref{tab:grids}).
A further finetune to $16\times192$ achieves $72.0\%$.
Whole-circuit optimization becomes ineffective on very long circuits.
For the $8\times512$ model, the improvement rate falls from $74.1\%$ on circuits of at most $32$ timesteps to $1.1\%$ at $257$--$512$ timesteps (Table~\ref{tab:bands}).
To handle circuits longer than its attention span, windowed mode divides them into $32$-timestep subcircuits, optimizes each independently, retains the original subcircuit whenever no improvement is found, merges the results, and applies one final optimization pass.
This raises the improvement rate to $98.4\%$ at $16\times192$.
However, it limits reduction because any window that the model cannot shorten is retained unchanged.
On circuits of $257$--$512$ timesteps, this limits BOPS to a $1.23\times$ reduction, while PyZX-fr reaches $18.0\times$ on the same sources (Table~\ref{tab:app-bands-red}).
Overall, BOPS can be finetuned for circuits with $2\times$ as many qubits and $3\times$ the depth at a fraction of the original training cost while maintaining strong reduction, although circuits exceeding its attention span remain its principal limitation.

\subsection{Cost}
\label{sec:cost}

BOPS trades higher inference cost for greater optimization reach.
A whole-circuit pass evaluates $64$ candidates over $128$ denoising steps and takes $4.8$\,s at $8\times64$ and $7.7$\,s at $8\times512$, while windowed mode increases these times to $14.4$\,s and $40.3$\,s.
The pattern-based optimizers take only $0.03$--$0.10$\,s, while QUESO and GUOQ average $5.09$\,s and $4.17$\,s under a $5$\,s budget (Table~\ref{tab:q8t64}).
The additional latency therefore buys access to optimizations that the cheaper methods do not find.

Unlike most baselines, BOPS can use a larger inference budget to find better circuits.
Its candidates are independent, so additional compute can evaluate more candidates or make additional passes over windows without retraining the model.
Table~\ref{tab:bench} therefore gives BOPS $10\times$ the budget used in Table~\ref{tab:q8t64}.
The deterministic baselines return the same circuit regardless of their budget, while QUESO and GUOQ use additional time to extend a sequential search.
GUOQ applies its expensive resynthesis step in only $1.5\%$ of iterations~\cite{guoq}.
A larger budget thus gives BOPS more independent opportunities to find an improvement, while other methods either cannot use the additional time or devote it to a single sequential search.

BOPS also has a predictable cost that can be amortized over many circuits.
Each whole-circuit pass uses $8{,}192$ denoiser evaluations regardless of circuit occupancy, although each evaluation becomes more expensive as the circuit grows.
Training is a one-off cost of $45$ GPU-hours at $8\times64$, followed by $7$ GPU-hours for the $16$-qubit finetune.
At $8\times64$, only $3.5\%$ of sources return without a verified improvement, while above ten qubits the checker limits the improvement rates that can be established (\S\ref{sec:verification}).
Notably, in Table~\ref{tab:grids}, the $8\times512$ windowed improvement rate is approximately that of the $8\times64$ whole-circuit mode: dividing denoiser evaluation into independent GPU-parallelized windows of $32$ circuit layers compensates for the $\sim\!10\times$ depth increase.
At scale, BOPS therefore pays a predictable inference cost without trading away correctness.

\section{Conclusion}
\label{sec:conclusion}

BOPS reformulates circuit optimization as learned transport between equivalent circuits.
From source--target pairs, it learns to generate optimized circuits without exposing the model to rewrite rules, gate algebra, or an equivalence loss.
The current model is trained on Clifford+$T$, but the formulation can be retargeted to another finite gate set or device by changing the token alphabet, constructing device-specific pairs, and retraining.
QUESO pursues related retargetability by synthesizing a rewrite-based optimizer for each device~\cite{queso}, whereas BOPS places this specialization in the data and reuses one learning architecture across targets.

Source-circuit conditioning makes generative optimization compatible with the interface and correctness requirements of a compiler.
Previous diffusion models synthesize circuits from encoded target unitaries~\cite{furrutter2024diffusion,furrutter2025multimodal}.
BOPS instead conditions on the valid source circuit already in the compiler, which compactly specifies the unitary and supplies an implementation to improve (\S\ref{sec:towards}).
Training requires neither unitary construction nor equivalence checking. At inference, every accepted candidate passes unitary equivalence verification; if none does, BOPS returns the source.
This source-conditioned transport turns diffusion synthesis into a compiler pass with a compact input and a verified output.

BOPS outperforms all evaluated hard-coded and search-based baselines on the rule-generated corpus and remains within their range on benchmark programs whose structure is absent from its training data.
To our knowledge, BOPS is the first generative optimizer to match or outperform state-of-the-art compiler pipelines for this task while verifying every reported output.

\paragraph{Future work.}
Future work should extend BOPS to new circuit families, larger instances, and additional compilation tasks while making better use of inference compute.
Windowed BOPS closes $95.2\%$ of the gap to its corpus targets on the $8\times64$ test split, leaving little of the target reduction to recover; new pairs should therefore include arithmetic, chemistry, and other structured circuits, with splits that keep every benchmark outside training.
The $16\times192$ results motivate direct training at greater width and depth, together with attention that spans the full circuit and equivalence verification that returns fewer undecided results.
Topology-aware routing is a direct circuit-to-circuit extension; stabilizer-tableau synthesis and fault-tolerant code search are more speculative and require new representations, objectives, and verifiers.
Recent inference-time scaling could replace independent samples with verifier-guided tree search, prioritizing denoising trajectories by expected verified gate reduction while retaining exact equivalence~\cite{guo2025treeg,zhang2026search}.
Together, these directions would extend BOPS into a verified, data-driven compiler across circuit families, hardware targets, and compilation stages.

\begin{acks}
    The authors acknowledge the use of resources provided by the Isambard-AI National AI Research Resource (AIRR). 
    Isambard-AI is operated by the University of Bristol and is funded by the UK Government’s Department for Science, Innovation and Technology (DSIT) via UK Research and Innovation; and the Science and Technology Facilities Council [ST/AIRR/I-A-I/1023].
    L.S.H. is supported by a Werner Siemens Fellowship of the Werner Siemens Foundation awarded by the Swiss Study Foundation.
    L.Y is supported through the Tencent Post-Doctoral Research Fellowship.
\end{acks}

\bibliographystyle{ACM-Reference-Format}
\bibliography{references}

\clearpage
\appendix

\section{Reproducibility}\label{app:reprod}

\subsection{Usage of the BOPS library}

BOPS ships as the \texttt{qcopt} package, with subcommands \texttt{train},
\texttt{evaluate}, \texttt{bench} and \texttt{corpus}.
Training and evaluation need only PyTorch and NumPy; the baseline optimizers are
an optional extra and each is loaded only when named.
A checkpoint is scored over a corpus split with

\begin{quote}\footnotesize
\begin{verbatim}
qcopt evaluate CKPT --data ROOT --part test
  --out DIR --split default --seed 0
  --num-candidates 64 --nfe 128 --rounds 1
  --transforms auto --buckets auto
  --precision bf16 --device cuda --tf32
\end{verbatim}
\end{quote}

\noindent
which draws $64$ candidates per source, spends $128$ denoising steps on each, and
scores the shortest one verified equivalent.
\texttt{--rounds 1} is a single shot: above one, the shortest verified candidate
becomes the next round's source.
\texttt{--transforms auto} restates a source before sampling --- identity,
conjugate, reverse, wire permutation, commutation slide --- taking turns across
the candidates; each is transformed back and judged against the untransformed
source, so a transform can spend a candidate but never invent a success.
Windowed runs add \texttt{--window-widths}, one width per peel round, each a
length the checkpoint trained on: the first round cuts the source into blocks and
merges the optimized blocks, and each later round re-cuts the previous merge.
A block that yields nothing keeps its own gates, so a merge is equivalent by
construction.
The $16$-qubit runs add \texttt{--verifier tiered --verify-workers 32
--verify-timeout 120}, selecting the check of \S\ref{sec:verification} and the
seconds one pair may take before it is called undecided; the flag is ignored at
$Q\le10$, where dense is complete and faster.
Vendored OpenQASM~2 circuits are scored by \texttt{qcopt bench} against the model
and the optimizer registry in one pass.
The accompanying code release documents the remaining flags, the corpus
generation pipeline, and the per-source and per-block records each run writes.
Training and inference use eight NVIDIA GH200 GPUs across two nodes; every
baseline runs on one core, one circuit at a time, so a per-circuit baseline cost
is a single-core cost.

\subsection{Experiment settings for external libraries}\label{app:extlib}

Every backend is reached through one conversion layer.
A gate list becomes a Qiskit \texttt{QuantumCircuit}, the backend rewrites it, and \texttt{transpile} rebases the result onto the gate pool of~\eqref{eq:gateset}, with the external
tools exchanging OpenQASM~2 at both ends.
A result is kept only when it decodes to the pool, fits the $Q\times D$ grid, and matches the source operator under the check of \S\ref{sec:verification}.
Below, \texttt{pool} is that gate set (Clifford+T), \texttt{s} the run seed, fixed at $0$,
\texttt{c} the source as a PyZX circuit and \texttt{g} its \texttt{to\_graph()}.

\paragraph{Qiskit 2.5.1.}
We run the transpiler at optimization level three, its strongest preset, with the
target basis fixed to our gate pool.
\begin{quote}\footnotesize
\begin{verbatim}
transpile(qc, basis_gates=pool,
          optimization_level=3, seed_transpiler=s)
\end{verbatim}
\end{quote}

\paragraph{PyZX 0.10.5.}
PyZX rewrites a circuit as a ZX diagram, and we run three configurations of it.
PyZX-b applies the peephole pass alone and never builds the diagram, PyZX-fr
reduces the diagram and extracts a circuit back out of it, and PyZX-fo ends that
same pipeline with \texttt{full\_optimize}, which adds the TODD phase-polynomial
pass~\cite{Heyfron2019todd}.
\begin{quote}\footnotesize
\begin{verbatim}
PyZX-b   zx.basic_optimization(c.to_basic_gates())
PyZX-fr  zx.full_reduce(g); g.normalize();
         zx.extract_circuit(g); then PyZX-b
PyZX-fo  PyZX-fr, with zx.full_optimize
         in place of zx.basic_optimization
\end{verbatim}
\end{quote}

\begingroup
\emergencystretch=1em
\paragraph{tket 2.18.1.}
We run the one pass that cancels and merges neighboring gates, because tket's
stronger passes answer in a continuous gate set and are not comparable here.
\texttt{RemoveRedundancies().apply(c)}.

\paragraph{VOQC (pyvoqc 0.1.1).}
VOQC applies a fixed list of Clifford+$T$ passes that propagate $X$ gates, reduce
Hadamards and cancel neighboring gates.
Three of them run twice, because canceling one kind of gate exposes the other.
\par
\endgroup
\begin{quote}\footnotesize
\begin{verbatim}
not_propagation, hadamard_reduction,
cancel_two_qubit_gates, cancel_single_qubit_gates,
cancel_two_qubit_gates, hadamard_reduction,
cancel_single_qubit_gates, replace_rzq
\end{verbatim}
\end{quote}

\paragraph{T$|$zap$\rangle$ 0.1.0, a pre-release build of July 2026.}
T$|$zap$\rangle$ is a standalone binary that optimizes $T$-count rather than
total gate count.
\texttt{tzap input.qasm -o output.qasm}.

\paragraph{QUESO and GUOQ (GUOQ-1.0).}
QUESO searches over rewrite rules, and GUOQ is that same search with numerical
resynthesis added.
One jar implements both, and \texttt{-resynth} is the only difference.
\begin{quote}\footnotesize
\begin{verbatim}
java -ea -cp JAR qoptimizer.Optimizer
  --rules-dir RULES -g CLIFFORDT -opt TOTAL
  -resynth NONE|BQSKIT --seed s
  -out . -job qcopt input.qasm
\end{verbatim}
\end{quote}

\noindent
GUOQ reaches BQSKit over a socket on port $8080$.
Each is given a $5$~s search budget per circuit and terminated at expiry, its
best solution so far read.
We lift the jars' default $100$-gate cap so both run on every circuit, and a
circuit a tool does not shorten is reported at its source length.

\section{Circuit training curriculum}\label{app:train}

\subsection{Motif library}
\label{app:motifs}
The generator of \S\ref{sec:reverse} draws targets from uniformly random gates,
from the motifs below, and from repeated copies of one motif.
Every motif is exactly representable in the Clifford$+T$ gate set
of~\eqref{eq:gateset}, which is the criterion by which they were selected.

\subsection{Rewrite-rule catalog}
\label{app:rules}
The $86$ rules used for reduction and for reverse rewriting (\S\ref{sec:reverse})
are listed below, grouped by the number of qubits they act on and by whether they
are built in or SAT-mined.
Each rule is drawn as its longer side, an equality, and its shorter side;
reduction applies a rule left to right and a reverse rewrite applies it right to
left.

\section{Asymptotic scaling}\label{app:asymp}

\subsection{Asymptotic scaling of problem size}\label{app:prob}
In this section, we derive two lemmas characterizing how the number of circuits representable by our model scales with the qubit count $n$, gate count $G$, and circuit depth $D$.
In this section, we refer to circuits by their syntactic representation in Figure~\ref{fig:representation}: an ordered sequence of gates together with their qubit locations, where each layer can have at most one gate per qubit.
Consequently, encodings that differ syntactically are counted separately, even if they implement the same unitary. This characterizes the size of the representation and search space of the model.

First, we consider the simpler case of $n$-qubit circuits consisting of an ordered sequence of $G$ gates, representing typical quantum programs where gates are appended one at a time.
\begin{lemma}\label{lem:countfixG}
    The number of $n$-qubit Clifford+T circuit representations of gate count $G$ is $2^{\Theta(G \log n)}$.
\end{lemma}
\begin{proof}
    The following holds for any finite gate set consisting of the CX gate and $q \geq 1$ one-qubit gates, such as the Clifford+T gate set.
    Taking into account positioning, there are $\mathcal{G} = n(n-1) + qn = n^2 + (q-1)n$ different gates on $n$ qubits.
    A circuit of gate count $G$ is an ordered sequence of $G$ such gates, so the number of distinct circuits is $\mathcal{G}^G$.
    As $n^2 \leq n^2 + (q-1)n \leq q n^2$ for all $n$, $\mathcal{G}^G$ is lower bounded by $2^{G (2 \log_2 n)}$ and upper bounded by $2^{G (2 \log_2 n + \log_2 q)}$.
    Because $q$ is constant, $2G \log_2 n + G \log_2 q = \Theta(G \log n)$, and hence $\mathcal{G}^G = 2^{\Theta(G \log n)}$.
\end{proof}

We next count $n$-qubit circuits of depth at most $D$ to determine how the number of circuits scales compared to the circuit spacetime cost $nD$.
\begin{lemma}\label{lem:countfixD}
    The number of $n$-qubit Clifford+T circuit representations of depth at most $D$ is $2^{\Theta(D n \log n)}$.
\end{lemma}
\begin{proof}
The following holds for any finite gate set consisting of the CX gate and $q$ one-qubit gates, such as the Clifford+T gate set.
All depth-one $n$-qubit circuits are accounted for by matching $2k$ qubits into $k$ disjoint pairs for all $k \leq \frac{n}{2}$, assigning to each pair one of the two CX orientations, and assigning to each of the remaining $n-2k$ qubits either one of the $q$ one-qubit gates or identity.
The total number of such depth-one circuits is:
\begin{align*}
     C_1 &= \sum_{k=0}^{\lfloor n/2 \rfloor} \binom{n}{2k} (2k-1)!! \; 2^k \; (q+1)^{n-2k}\\
     &= \sum_{k=0}^{\lfloor n/2 \rfloor} \binom{n}{2k} \frac{(2k)!}{2^k k!} \; 2^k \; (q+1)^{n-2k}
\end{align*}

To upper bound $C_1$, apply $\binom{n}{2k}\frac{(2k)!}{2^k k!}= \binom{n}{k} \frac{\prod_{j=0}^{k-1}(n-k-j)}{2^k}$:
\begin{align*}
C_1 &\leq (q+1)^n \sum_{k=0}^{\lfloor n/2 \rfloor} \binom{n}{k} \left(\frac{n}{2}\right)^k 2^k \\
&\leq (q+1)^n \sum_{k=0}^{\lfloor n/2 \rfloor} \binom{n}{k} n^k \\
&\leq (q+1)^n n^\frac{n}{2} \sum_{k=0}^{\lfloor n/2 \rfloor} \binom{n}{k} \\
&\leq (q+1)^n n^\frac{n}{2} 2^n \text{ \emph{(binomial theorem)}} = 2^{O(n \log n)}
\end{align*}

To lower bound $C_1$:
\begin{align*}
C_1 &\geq \sum_{k=\lfloor n/2 \rfloor}^{\lfloor n/2 \rfloor} \frac{n!}{(n - 2k)! \; 2^k k!} \; 2^k \\
&\geq \frac{n!}{2^{\lfloor n/2 \rfloor} \lfloor n/2 \rfloor!} \; 2^{\lfloor n/2 \rfloor} \text{ \emph{(since $(n-2\lfloor n/2 \rfloor)! = 1$)}} \\
&\geq \frac{\sqrt{2\pi n} \left(\frac{n}{e}\right)^n}{e \sqrt{\lfloor n/2 \rfloor} \left(\frac{\lfloor n/2 \rfloor}{e}\right)^{\lfloor n/2 \rfloor}} \text{ \emph{(Stirling bounds)}} \\
&\geq \frac{\sqrt{2\pi n} \left(\frac{n}{e}\right)^n}{e \sqrt{n/2} \left(\frac{n}{2e}\right)^{n/2}} = \frac{2\sqrt{\pi}}{e} \left(\frac{2n}{e}\right)^{n/2} = 2^{\Omega(n \log n)}
\end{align*}
where Stirling bounds used are $x! \ge \sqrt{2\pi x}(x/e)^x$ for $x = n$ and $x! \le e \sqrt{x} (x/e)^x$ for $x = \lfloor n/2 \rfloor$, and the last line uses $\lfloor n/2 \rfloor \leq \frac{n}{2}$.

Thus, there are $(C_1)^D = 2^{\Theta(D n \log n)}$ $n$-qubit Clifford+T circuits of depth at most $D$ having distinct circuit representations.
\end{proof}

\subsection{Asymptotic scaling of equivalence checking}\label{app:equiv}

In \S\ref{sec:verification}, we discuss two tests for equivalence of Clifford+T circuits, which answers positively only for equivalent circuits (i.e. no false positives), but has a small probability of not definitively deciding whether the two circuits are equivalent.
In this section, we summarize the approach of strong classical simulation of quantum computation, as a third test for equivalence which answers definitively, and recount how it scales with circuit size.
Strong simulation computes exact amplitudes of measurement outcomes, in contrast to weak simulation which samples measurement outcomes from the correct probability distribution.
Strong simulation implies weak simulation.

Equivalence checking follows directly from strong simulation (but not weak simulation) and is even more efficient than computing the full unitary matrix.
This is because to check whether $W = U(C)^\dagger U(C')$ is the identity matrix up to a global phase, it suffices to compute just the diagonal $\langle x|W|x\rangle$ which should all equal the same norm-1 complex number: One entry different from the rest refutes the equivalence.

In constrast to statevector simulation methods, state-of-the-art classical simulation approaches through stabilizer decompositions instead split the simulation task into that of simulating exponentially many Clifford terms, which are then summed over.
Each Clifford term's simulation complexity is polynomial in the size of its labeled graph.
These terms are simultaneously simulable on parallel threads, and so are most runtime-efficient when the number of available threads is at least the number of terms.

\begin{lemma}\label{lem:circgraph}
    Any $n$-qubit depth-$D$ Clifford+T circuit, converted gate-by-gate into the ZX-calculus~\cite{ZX}, is a labeled graph consisting of $O(n D)$ vertices and $O(n D)$ edges.
\end{lemma}
\begin{proof}
     Clifford+T circuits as ZX-calculus diagrams have two types of edges (identity or Hadamard), and two types of vertices (Z or X). Each vertex has 8 possible labels, corresponding to the 8 possible integer multiples of $\frac{\pi}{4}$ modulo $2\pi$.
     Consider starting with $n$ identity edges and adding gate by gate: $H$ gates toggle an identity edge to a Hadamard edge, $CX$ gates add two vertices and three edges, and $\{S, S^\dagger, T, T^\dagger\}$ all add one vertex and one edge.
     The circuit therefore has at most $n D$ vertices (as each position coordinate adds at most one vertex) and at most $n (D+1) + \frac{n D}{2}$ edges (saturated by $\frac{n D}{2}$ CX gates, on $n (D+1)$ wire segments including boundary edges).
\end{proof}

The general approach of stabilizer decompositions exhibits the following asymptotic scaling, where in practice performance can scale with smaller effective $\alpha$ based on the efficiency of the decompositions applicable to the circuit in question.
\begin{lemma}\label{lem:clsimcost}
    An $n$-qubit, depth-$D$ Clifford+T circuit with $t$ $T$ gates is classically simulable at an asymptotic cost of $O(2^{\alpha t})$ Clifford terms for some constant $\alpha$, each describable by a ZX-calculus labeled graph consisting of $O(n D)$ vertices and $O(n D)$ edges.
\end{lemma}
\begin{proof}
    Standard stabilizer decompositions decompose a constant number of non-Clifford vertices (i.e. any vertices labeled by any odd multiple of $\frac{\pi}{4}$) at a time, into a sum over a constant number of Clifford graphs, each of a constant size.
    Doing this to the initial labeled graph of $O(n D)$ vertices and edges per Lemma~\ref{lem:circgraph}, for all $t$ $T$ gates, results in a number of Clifford graphs exponential in $t$, where each Clifford graph has $O(n D)$ vertices and edges.
\end{proof}

The best-known upper bound for exact stabilizer decomposition is $\alpha \approx 0.396$~\cite{Qassim2021alpha} which is achievable for circuits of any size~\cite{kissinger2022stabdecomp}.
A state-of-the-art GPU-parallelized stabilizer decompositions implementation called ParamZX~\cite{sutcliffe2024paramzx} is accessible from PyZX. Although BOPS used PyZX without ParamZX in the reported results of this work for equivalence checking, BOPS supports using ParamZX which is empirically 3-4 orders of magnitude faster than PyZX.
These reports suggest that using ParamZX for definitive equivalence checking would be expected to comprise an inconsequential fraction of the total generative inference time.

\clearpage

\section{Full result tables}
\label{app:tables}

\begin{table}[htbp]
\caption{Each optimizer on its own rather than the per-circuit best, on the $68$ benchmark circuits of Table~\ref{tab:bench} and the $119$ of Table~\ref{tab:bench512}.
A depth below $1.00\times$ is a set made deeper while gates came out.}
\label{tab:app-bench-per-opt}
\footnotesize
\setlength{\tabcolsep}{4pt}
\centering
\begin{tabular}{@{}lrrrr@{}}
\toprule
& \multicolumn{2}{c}{$8\times64$} & \multicolumn{2}{c}{$8\times512$} \\
\cmidrule(lr){2-3}\cmidrule(lr){4-5}
system & gates & depth & gates & depth \\
\midrule
Qiskit~O3        & $1.09\times$ & $1.06\times$ & $1.10\times$ & $1.05\times$ \\
PyZX-b           & $1.10\times$ & $1.03\times$ & $1.12\times$ & $1.03\times$ \\
PyZX-fr          & $1.11\times$ & $1.04\times$ & $1.15\times$ & $1.00\times$ \\
PyZX-fo          & $1.13\times$ & $\boldsymbol{1.08\times}$ & $1.11\times$ & $1.06\times$ \\
tket             & $1.08\times$ & $1.06\times$ & $1.07\times$ & $1.04\times$ \\
VOQC             & $1.09\times$ & $1.07\times$ & $1.10\times$ & $1.06\times$ \\
T$|$zap$\rangle$ & $1.15\times$ & $\boldsymbol{1.08\times}$ & $1.19\times$ & $\boldsymbol{1.08\times}$ \\
QUESO            & $\boldsymbol{1.26\times}$ & $1.06\times$ & $\boldsymbol{1.22\times}$ & $1.00\times$ \\
GUOQ             & $1.17\times$ & $0.99\times$ & $1.04\times$ & $1.00\times$ \\
\midrule
BOPS             & $1.13\times$ & $\boldsymbol{1.08\times}$ & $1.09\times$ & $1.04\times$ \\
\bottomrule
\end{tabular}
\end{table}

\begin{table}[htbp]
\caption{The seven optimizers run on \texttt{q8t512}, by source length, gates reduced by, on the same bands as Table~\ref{tab:bands}.
QUESO and GUOQ were not run on this set.}
\label{tab:app-bands-red}
\footnotesize
\setlength{\tabcolsep}{4pt}
\centering
\begin{tabular}{@{}lrrrrrr@{}}
\toprule
& \multicolumn{5}{c}{source columns} & \\
\cmidrule(lr){2-6}
system & $\le32$ & $33$--$64$ & $65$--$128$ & $129$--$256$ & $257$--$512$ & all \\
\emph{circuits}  & 1{,}154 & 452 & 650 & 354 & 371 & 2{,}981 \\
\midrule
  Qiskit~O3 & $1.20\times$ & $1.25\times$ & $1.29\times$ & $1.30\times$ & $1.32\times$ & $1.25\times$ \\
  PyZX-b & $1.51\times$ & $1.99\times$ & $2.30\times$ & $2.33\times$ & $2.96\times$ & $1.97\times$ \\
  PyZX-fr & $1.56\times$ & $\boldsymbol{2.56\times}$ & $\boldsymbol{4.35\times}$ & $\boldsymbol{7.06\times}$ & $\boldsymbol{18.0\times}$ & $\boldsymbol{3.41\times}$ \\
  PyZX-fo & $1.49\times$ & $2.40\times$ & $4.02\times$ & $6.60\times$ & $17.5\times$ & $3.23\times$ \\
  tket & $1.09\times$ & $1.11\times$ & $1.14\times$ & $1.16\times$ & $1.16\times$ & $1.12\times$ \\
  VOQC & $1.23\times$ & $1.31\times$ & $1.36\times$ & $1.38\times$ & $1.43\times$ & $1.31\times$ \\
  T$|$zap$\rangle$ & $1.25\times$ & $1.32\times$ & $1.38\times$ & $1.41\times$ & $1.46\times$ & $1.33\times$ \\
\midrule
  BOPS, whole & $1.67\times$ & $1.66\times$ & $1.44\times$ & $1.07\times$ & $1.04\times$ & $1.45\times$ \\
  BOPS, windowed & $\boldsymbol{1.84\times}$ & $2.25\times$ & $1.92\times$ & $1.29\times$ & $1.23\times$ & $1.75\times$ \\
\bottomrule
\end{tabular}
\end{table}

\begin{table}[htbp]
\caption{\texttt{q8t512} by source length.
The seven optimizers run on this set are in Table~\ref{tab:app-bands-red}.}
\label{tab:bands}
\footnotesize
\setlength{\tabcolsep}{4pt}
\centering
\begin{tabular}{@{}lrrrrrr@{}}
\toprule
& \multicolumn{5}{c}{source columns} & \\
\cmidrule(lr){2-6}
& $\le32$ & $33$--$64$ & $65$--$128$ & $129$--$256$ & $257$--$512$ & all \\
\emph{circuits}  & 1{,}154 & 452 & 650 & 354 & 371 & 2{,}981 \\
\midrule
\multicolumn{7}{@{}l}{\emph{gates reduced by}} \\
  \quad whole & $1.67\times$ & $1.66\times$ & $1.44\times$ & $1.07\times$ & $1.04\times$ & $1.45\times$ \\
  \quad windowed & $\boldsymbol{1.84\times}$ & $\boldsymbol{2.25\times}$ & $\boldsymbol{1.92\times}$ & $\boldsymbol{1.29\times}$ & $\boldsymbol{1.23\times}$ & $\boldsymbol{1.75\times}$ \\
\midrule
\multicolumn{7}{@{}l}{\emph{improved}} \\
  \quad whole & 74.1\% & 41.6\% & 17.4\% & 2.5\% & 1.1\% & 39.2\% \\
  \quad windowed & \textbf{84.8\%} & \textbf{80.8\%} & \textbf{82.6\%} & \textbf{79.1\%} & \textbf{88.7\%} & \textbf{83.5\%} \\
\bottomrule
\end{tabular}
\end{table}

\begin{table}[htbp]
\caption{Target length in gates on the rule-generated part of each test split -- \emph{atomic}, \emph{few}, \emph{medium} and \emph{hard} pooled, the same circuits scored in Table~\ref{tab:grids}.}
\label{tab:app-target-test}
\footnotesize
\setlength{\tabcolsep}{4pt}
\centering
\begin{tabular}{@{}lrr@{}}
\toprule
dataset & mean & median \\
\midrule
\texttt{q8t64} ($8\times64$)     & 16.55 & 15 \\
\texttt{q16t64} ($16\times64$)   & 29.80 & 26 \\
\texttt{q16t192} ($16\times192$) & 33.72 & 27 \\
\texttt{q8t512} ($8\times512$)   & 21.10 & 16 \\
\bottomrule
\end{tabular}
\end{table}

\begin{table}[htbp]
\caption{As Table~\ref{tab:app-target-test}, over the training splits, counting each generated pair once rather than once per trajectory-derived row.}
\footnotesize
\setlength{\tabcolsep}{4pt}
\centering
\begin{tabular}{@{}lrr@{}}
\toprule
dataset & mean & median \\
\midrule
\texttt{q8t64} ($8\times64$)     & 16.25 & 15 \\
\texttt{q16t64} ($16\times64$)   & 29.81 & 26 \\
\texttt{q16t192} ($16\times192$) & 33.05 & 27 \\
\texttt{q8t512} ($8\times512$)   & 20.06 & 15 \\
\bottomrule
\end{tabular}
\end{table}

\begin{table*}[p]
\caption{Full gate-count and depth results for the benchmark circuits summarized in Table~\ref{tab:bench}. Reductions are bold where BOPS matches or beats all nine systems. Each \emph{best} is the best of the nine on that circuit alone, so its gate and depth columns need not name the same optimizer. \emph{src}: source gates; \emph{cut}: gates the best system removes. \emph{in} names the suites: \textsuperscript{F}Feynman, \textsuperscript{G}GUOQ, \textsuperscript{B}QASMBench, and \textsuperscript{Q}QUESO. Optimizer superscripts are \textsuperscript{qi}Qiskit~O3, \textsuperscript{pb}PyZX-b, \textsuperscript{pr}PyZX-fr, \textsuperscript{po}PyZX-fo, \textsuperscript{tk}tket, \textsuperscript{vq}VOQC, \textsuperscript{tz}T$|$zap$\rangle$, \textsuperscript{qu}QUESO, and \textsuperscript{gq}GUOQ. A system returning nothing strictly shorter is scored on its input.}
\label{tab:bench-full}
\fontsize{8.2pt}{9pt}\selectfont
\setlength{\tabcolsep}{3pt}
\centering
\begin{tabular}{@{}lcrrrrrr@{}}
\toprule
& & & & \multicolumn{2}{c}{gates} & \multicolumn{2}{c}{depth} \\
\cmidrule(lr){5-6}\cmidrule(lr){7-8}
circuit & in & src & cut & best & BOPS & best & BOPS \\
\midrule
  lpn\_n5 & \textsuperscript{B} & 11 & 8 & \textsuperscript{pb}\,$3.67\times$ & $\boldsymbol{3.67\times}$ & \textsuperscript{pb}\,$1.33\times$ & $\boldsymbol{1.33\times}$ \\
  bb84\_n8 & \textsuperscript{B} & 54 & 34 & \textsuperscript{qi}\,$2.70\times$ & $\boldsymbol{3.00\times}$ & \textsuperscript{qi}\,$1.57\times$ & $\boldsymbol{1.57\times}$ \\
  qec\_en\_n5 & \textsuperscript{B} & 25 & 15 & \textsuperscript{qu}\,$2.50\times$ & $2.27\times$ & \textsuperscript{qu}\,$2.12\times$ & $1.70\times$ \\
  hs4\_n4 & \textsuperscript{B} & 40 & 24 & \textsuperscript{qu}\,$2.50\times$ & $2.00\times$ & \textsuperscript{qu}\,$2.00\times$ & $1.56\times$ \\
  iswap\_n2 & \textsuperscript{B} & 12 & 4 & \textsuperscript{qu}\,$1.50\times$ & $\boldsymbol{1.50\times}$ & \textsuperscript{qu}\,$1.43\times$ & $\boldsymbol{1.43\times}$ \\
  rd32-v1\_68 & \textsuperscript{G} & 42 & 21 & \textsuperscript{qu}\,$2.00\times$ & $1.45\times$ & \textsuperscript{qu}\,$2.25\times$ & $1.42\times$ \\
  deutsch\_n2 & \textsuperscript{B} & 8 & 2 & \textsuperscript{qi}\,$1.33\times$ & $\boldsymbol{1.33\times}$ & \textsuperscript{qi}\,$1.40\times$ & $\boldsymbol{1.40\times}$ \\
  rd32-v0\_66 & \textsuperscript{G} & 34 & 15 & \textsuperscript{qu}\,$1.79\times$ & $1.31\times$ & \textsuperscript{pb}\,$1.54\times$ & $1.25\times$ \\
  mod5\_4 & \textsuperscript{FGQ} & 66 & 29 & \textsuperscript{po}\,$1.78\times$ & $1.25\times$ & \textsuperscript{po}\,$1.70\times$ & $1.08\times$ \\
  ham3\_102 & \textsuperscript{GQ} & 20 & 4 & \textsuperscript{qu}\,$1.25\times$ & $1.18\times$ & \textsuperscript{qi}\,$1.18\times$ & $\boldsymbol{1.18\times}$ \\
  barenco\_tof\_3 & \textsuperscript{FGQ} & 60 & 16 & \textsuperscript{qu}\,$1.36\times$ & $1.15\times$ & \textsuperscript{---}\,$1.00\times$ & $\boldsymbol{1.00\times}$ \\
  4mod5-v1\_22 & \textsuperscript{GQ} & 24 & 4 & \textsuperscript{qu}\,$1.20\times$ & $1.14\times$ & \textsuperscript{qi}\,$1.15\times$ & $\boldsymbol{1.25\times}$ \\
  tof\_4 & \textsuperscript{FGQ} & 75 & 15 & \textsuperscript{qu}\,$1.25\times$ & $1.12\times$ & \textsuperscript{tz}\,$1.06\times$ & $1.03\times$ \\
  qiskit-alu-v0\_26 & \textsuperscript{G} & 87 & 17 & \textsuperscript{qu}\,$1.24\times$ & $1.12\times$ & \textsuperscript{tz}\,$1.08\times$ & $1.04\times$ \\
  decod24-v0\_38 & \textsuperscript{G} & 54 & 10 & \textsuperscript{qu}\,$1.23\times$ & $1.10\times$ & \textsuperscript{qu}\,$1.14\times$ & $1.10\times$ \\
  alu-bdd\_288 & \textsuperscript{G} & 87 & 26 & \textsuperscript{qu}\,$1.43\times$ & $1.10\times$ & \textsuperscript{qu}\,$1.31\times$ & $1.04\times$ \\
  alu-v0\_26 & \textsuperscript{G} & 87 & 17 & \textsuperscript{qu}\,$1.24\times$ & $1.10\times$ & \textsuperscript{tz}\,$1.08\times$ & $1.04\times$ \\
  barenco\_tof\_4 & \textsuperscript{FGQ} & 114 & 28 & \textsuperscript{tz}\,$1.33\times$ & $1.10\times$ & \textsuperscript{tz}\,$1.05\times$ & $1.02\times$ \\
  4gt13\_91 & \textsuperscript{G} & 103 & 25 & \textsuperscript{qu}\,$1.32\times$ & $1.10\times$ & \textsuperscript{---}\,$1.00\times$ & $\boldsymbol{1.00\times}$ \\
  4gt13-v1\_93 & \textsuperscript{G} & 74 & 12 & \textsuperscript{tz}\,$1.19\times$ & $1.09\times$ & \textsuperscript{tz}\,$1.02\times$ & $1.00\times$ \\
  4mod5-v0\_18 & \textsuperscript{G} & 75 & 32 & \textsuperscript{qu}\,$1.74\times$ & $1.09\times$ & \textsuperscript{qu}\,$1.48\times$ & $1.00\times$ \\
  decod24-bdd\_294 & \textsuperscript{G} & 88 & 37 & \textsuperscript{po}\,$1.73\times$ & $1.09\times$ & \textsuperscript{po}\,$1.33\times$ & $1.00\times$ \\
  alu-v0\_27 & \textsuperscript{GQ} & 39 & 4 & \textsuperscript{qu}\,$1.11\times$ & $1.08\times$ & \textsuperscript{qi}\,$1.09\times$ & $\boldsymbol{1.14\times}$ \\
  4mod5-bdd\_287 & \textsuperscript{G} & 79 & 18 & \textsuperscript{gq}\,$1.30\times$ & $1.08\times$ & \textsuperscript{qu}\,$1.18\times$ & $1.00\times$ \\
  rd32\_270 & \textsuperscript{G} & 93 & 46 & \textsuperscript{po}\,$1.98\times$ & $1.08\times$ & \textsuperscript{po}\,$1.85\times$ & $1.00\times$ \\
  alu-v3\_35 & \textsuperscript{GQ} & 40 & 4 & \textsuperscript{qu}\,$1.11\times$ & $1.08\times$ & \textsuperscript{qi}\,$1.09\times$ & $\boldsymbol{1.14\times}$ \\
  alu-v4\_37 & \textsuperscript{GQ} & 40 & 4 & \textsuperscript{qu}\,$1.11\times$ & $1.08\times$ & \textsuperscript{qi}\,$1.09\times$ & $\boldsymbol{1.14\times}$ \\
  4gt11\_82 & \textsuperscript{Q} & 27 & 3 & \textsuperscript{qu}\,$1.12\times$ & $1.08\times$ & \textsuperscript{qi}\,$1.18\times$ & $\boldsymbol{1.18\times}$ \\
  mod5d2\_64 & \textsuperscript{G} & 56 & 11 & \textsuperscript{po}\,$1.24\times$ & $1.08\times$ & \textsuperscript{po}\,$1.35\times$ & $1.00\times$ \\
  alu-v1\_29 & \textsuperscript{GQ} & 43 & 2 & \textsuperscript{qi}\,$1.05\times$ & $\boldsymbol{1.07\times}$ & \textsuperscript{pb}\,$1.17\times$ & $1.12\times$ \\
  alu-v2\_33 & \textsuperscript{GQ} & 43 & 2 & \textsuperscript{qi}\,$1.05\times$ & $\boldsymbol{1.07\times}$ & \textsuperscript{qi}\,$1.08\times$ & $\boldsymbol{1.12\times}$ \\
  4mod5-v1\_23 & \textsuperscript{G} & 72 & 38 & \textsuperscript{po}\,$2.12\times$ & $1.07\times$ & \textsuperscript{po}\,$2.20\times$ & $1.07\times$ \\
  4gt5\_76 & \textsuperscript{G} & 91 & 19 & \textsuperscript{qu}\,$1.26\times$ & $1.07\times$ & \textsuperscript{qu}\,$1.04\times$ & $1.02\times$ \\
  simon\_n6 & \textsuperscript{B} & 62 & 48 & \textsuperscript{qu}\,$4.43\times$ & $1.07\times$ & \textsuperscript{qu}\,$5.00\times$ & $1.06\times$ \\
  4gt13\_92 & \textsuperscript{GQ} & 66 & 13 & \textsuperscript{gq}\,$1.25\times$ & $1.06\times$ & \textsuperscript{qu}\,$1.06\times$ & $1.00\times$ \\
  alu-v3\_34 & \textsuperscript{GQ} & 55 & 6 & \textsuperscript{qu}\,$1.12\times$ & $1.06\times$ & \textsuperscript{qi}\,$1.06\times$ & $\boldsymbol{1.10\times}$ \\
  mod5mils\_65 & \textsuperscript{GQ} & 38 & 9 & \textsuperscript{gq}\,$1.31\times$ & $1.06\times$ & \textsuperscript{pr}\,$1.09\times$ & $1.00\times$ \\
  4mod5-v0\_19 & \textsuperscript{GQ} & 38 & 5 & \textsuperscript{tz}\,$1.15\times$ & $1.06\times$ & \textsuperscript{tz}\,$1.09\times$ & $1.00\times$ \\
  mod5\_4\_\_o3 & \textsuperscript{F} & 40 & 12 & \textsuperscript{gq}\,$1.43\times$ & $1.05\times$ & \textsuperscript{qu}\,$1.74\times$ & $1.06\times$ \\
  alu-v1\_28 & \textsuperscript{GQ} & 40 & 4 & \textsuperscript{qu}\,$1.11\times$ & $1.05\times$ & \textsuperscript{qi}\,$1.09\times$ & $\boldsymbol{1.09\times}$ \\
  decod24-v2\_43 & \textsuperscript{G} & 61 & 5 & \textsuperscript{tz}\,$1.09\times$ & $1.05\times$ & \textsuperscript{qi}\,$1.06\times$ & $\boldsymbol{1.09\times}$ \\
  tof\_3 & \textsuperscript{FGQ} & 45 & 10 & \textsuperscript{qu}\,$1.29\times$ & $1.05\times$ & \textsuperscript{---}\,$1.00\times$ & $\boldsymbol{1.00\times}$ \\
  toffoli\_n3 & \textsuperscript{B} & 24 & 1 & \textsuperscript{vq}\,$1.04\times$ & $\boldsymbol{1.04\times}$ & \textsuperscript{---}\,$1.00\times$ & $\boldsymbol{1.00\times}$ \\
  rd53\_138 & \textsuperscript{G} & 132 & 23 & \textsuperscript{gq}\,$1.21\times$ & $1.04\times$ & \textsuperscript{tz}\,$1.02\times$ & $1.00\times$ \\
  decod24-v1\_41 & \textsuperscript{G} & 91 & 14 & \textsuperscript{qu}\,$1.18\times$ & $1.03\times$ & \textsuperscript{tz}\,$1.02\times$ & $1.00\times$ \\
  barenco\_tof\_4\_\_o3 & \textsuperscript{F} & 68 & 6 & \textsuperscript{qu}\,$1.10\times$ & $1.03\times$ & \textsuperscript{qu}\,$1.05\times$ & $1.02\times$ \\
  cat\_state\_n4 & \textsuperscript{B} & 4 & 0 & \textsuperscript{---}\,$1.00\times$ & $\boldsymbol{1.00\times}$ & \textsuperscript{---}\,$1.00\times$ & $\boldsymbol{1.00\times}$ \\
  ex-1\_166 & \textsuperscript{GQ} & 22 & 0 & \textsuperscript{---}\,$1.00\times$ & $\boldsymbol{1.00\times}$ & \textsuperscript{---}\,$1.00\times$ & $\boldsymbol{1.00\times}$ \\
  ex1\_226 & \textsuperscript{G} & 13 & 0 & \textsuperscript{---}\,$1.00\times$ & $\boldsymbol{1.00\times}$ & \textsuperscript{---}\,$1.00\times$ & $\boldsymbol{1.00\times}$ \\
  graycode6\_47 & \textsuperscript{G} & 5 & 0 & \textsuperscript{---}\,$1.00\times$ & $\boldsymbol{1.00\times}$ & \textsuperscript{---}\,$1.00\times$ & $\boldsymbol{1.00\times}$ \\
  qrng\_n4 & \textsuperscript{B} & 4 & 0 & \textsuperscript{---}\,$1.00\times$ & $\boldsymbol{1.00\times}$ & \textsuperscript{---}\,$1.00\times$ & $\boldsymbol{1.00\times}$ \\
  teleportation\_n3 & \textsuperscript{B} & 8 & 0 & \textsuperscript{---}\,$1.00\times$ & $\boldsymbol{1.00\times}$ & \textsuperscript{---}\,$1.00\times$ & $\boldsymbol{1.00\times}$ \\
  tof\_3\_\_o3 & \textsuperscript{F} & 35 & 0 & \textsuperscript{---}\,$1.00\times$ & $\boldsymbol{1.00\times}$ & \textsuperscript{---}\,$1.00\times$ & $\boldsymbol{1.00\times}$ \\
  tof\_4\_\_o3 & \textsuperscript{F} & 55 & 0 & \textsuperscript{---}\,$1.00\times$ & $\boldsymbol{1.00\times}$ & \textsuperscript{---}\,$1.00\times$ & $\boldsymbol{1.00\times}$ \\
\bottomrule
\end{tabular}
\end{table*}

\begin{table*}[p]
\caption{As Table~\ref{tab:bench}, at $8\times512$: the $119$ vendored circuits that fit that grid, showing the $87$ BOPS shortens and the $8$ no system shortens.
BOPS is bold where it matches or beats all nine optimizers. \emph{cut}: gates that best removes.
The $24$ it leaves unchanged while some optimizer shortens them are \texttt{3\_17\_13}, \texttt{4gt11\_83}, \texttt{4gt11\_84}, \texttt{4mod5-v0\_18}, \texttt{4mod5-v0\_19}, \texttt{4mod5-v0\_20}, \texttt{4mod5-v1\_23}, \texttt{4mod5-v1\_24}, \texttt{adder\_n4}, \texttt{barenco\_tof\_3\_\_o3}, \texttt{bb84\_n8}, \texttt{decod24-bdd\_294}, \texttt{decod24-v2\_43}, \texttt{error\_correctiond3\_n5}, \texttt{fredkin\_n3}, \texttt{grover\_n2}, \texttt{mod5\_4}, \texttt{mod5d1\_63}, \texttt{one-two-three-v3\_101}, \texttt{qft\_4\_\_o3}, \texttt{qiskit-3\_17\_13}, \texttt{rd32-v1\_68}, \texttt{rd32\_270}, \texttt{tof\_4}.}
\label{tab:bench512}
\fontsize{8.2pt}{8.8pt}\selectfont
\setlength{\tabcolsep}{1.5pt}
\centering
\begin{tabular}{@{}lclrrr@{\hspace{1.5em}}lclrrr@{}}
\toprule
circuit & in & src & cut & best & BOPS & circuit & in & src & cut & best & BOPS \\
\midrule
lpn\_n5 & \textsuperscript{B} & 11 & 8 & \textsuperscript{pb}\,$3.67\times$ & $\boldsymbol{3.67\times}$ & alu-v2\_33 & \textsuperscript{GQ} & 43 & 2 & \textsuperscript{qi}\,$1.05\times$ & $\boldsymbol{1.07\times}$ \\
hs4\_n4 & \textsuperscript{B} & 40 & 24 & \textsuperscript{qu}\,$2.50\times$ & $\boldsymbol{2.50\times}$ & mini-alu\_167 & \textsuperscript{G} & 288 & 61 & \textsuperscript{gq}\,$1.27\times$ & $1.07\times$ \\
qec\_en\_n5 & \textsuperscript{B} & 25 & 15 & \textsuperscript{qu}\,$2.50\times$ & $2.27\times$ & mod10\_171 & \textsuperscript{G} & 247 & 41 & \textsuperscript{pr}\,$1.20\times$ & $1.07\times$ \\
iswap\_n2 & \textsuperscript{B} & 12 & 4 & \textsuperscript{qu}\,$1.50\times$ & $\boldsymbol{1.50\times}$ & mod5adder\_127 & \textsuperscript{G} & 585 & 90 & \textsuperscript{tz}\,$1.18\times$ & $1.07\times$ \\
rd32-v0\_66 & \textsuperscript{G} & 34 & 15 & \textsuperscript{qu}\,$1.79\times$ & $1.36\times$ & sf\_274 & \textsuperscript{G} & 820 & 223 & \textsuperscript{pr}\,$1.37\times$ & $1.07\times$ \\
deutsch\_n2 & \textsuperscript{B} & 8 & 2 & \textsuperscript{qi}\,$1.33\times$ & $\boldsymbol{1.33\times}$ & alu-v2\_31 & \textsuperscript{G} & 454 & 77 & \textsuperscript{tz}\,$1.20\times$ & $1.07\times$ \\
barenco\_tof\_3 & \textsuperscript{FGQ} & 60 & 16 & \textsuperscript{qu}\,$1.36\times$ & $1.30\times$ & 4gt5\_76 & \textsuperscript{G} & 91 & 19 & \textsuperscript{qu}\,$1.26\times$ & $1.07\times$ \\
tof\_3 & \textsuperscript{FGQ} & 45 & 10 & \textsuperscript{qu}\,$1.29\times$ & $1.22\times$ & 4gt4-v0\_72 & \textsuperscript{G} & 261 & 46 & \textsuperscript{tz}\,$1.21\times$ & $1.07\times$ \\
simon\_n6 & \textsuperscript{B} & 62 & 48 & \textsuperscript{qu}\,$4.43\times$ & $1.22\times$ & mod10\_176 & \textsuperscript{G} & 181 & 30 & \textsuperscript{qu}\,$1.20\times$ & $1.06\times$ \\
rd53\_131 & \textsuperscript{G} & 520 & 135 & \textsuperscript{pr}\,$1.35\times$ & $1.16\times$ & 4\_49\_16 & \textsuperscript{G} & 220 & 31 & \textsuperscript{qu}\,$1.16\times$ & $1.06\times$ \\
decod24-v0\_38 & \textsuperscript{G} & 54 & 11 & \textsuperscript{gq}\,$1.26\times$ & $1.15\times$ & barenco\_tof\_4\_\_o3 & \textsuperscript{F} & 68 & 6 & \textsuperscript{qu}\,$1.10\times$ & $1.06\times$ \\
4mod5-v1\_22 & \textsuperscript{GQ} & 24 & 4 & \textsuperscript{qu}\,$1.20\times$ & $1.14\times$ & mod8-10\_177 & \textsuperscript{G} & 443 & 62 & \textsuperscript{tz}\,$1.16\times$ & $1.06\times$ \\
rd53\_133 & \textsuperscript{G} & 580 & 136 & \textsuperscript{tz}\,$1.31\times$ & $1.12\times$ & mod8-10\_178 & \textsuperscript{G} & 345 & 48 & \textsuperscript{tz}\,$1.16\times$ & $1.06\times$ \\
mod5mils\_65 & \textsuperscript{GQ} & 38 & 7 & \textsuperscript{qu}\,$1.23\times$ & $1.12\times$ & 4gt12-v0\_86 & \textsuperscript{G} & 251 & 41 & \textsuperscript{tz}\,$1.20\times$ & $1.06\times$ \\
one-two-three-v1\_99 & \textsuperscript{G} & 135 & 23 & \textsuperscript{tz}\,$1.21\times$ & $1.12\times$ & one-two-three-v0\_97 & \textsuperscript{G} & 290 & 40 & \textsuperscript{tz}\,$1.16\times$ & $1.06\times$ \\
alu-v0\_26 & \textsuperscript{G} & 87 & 17 & \textsuperscript{qu}\,$1.24\times$ & $1.12\times$ & alu-v3\_34 & \textsuperscript{GQ} & 55 & 6 & \textsuperscript{qu}\,$1.12\times$ & $1.06\times$ \\
qiskit-alu-v0\_26 & \textsuperscript{G} & 87 & 17 & \textsuperscript{qu}\,$1.24\times$ & $1.12\times$ & 4gt4-v1\_74 & \textsuperscript{G} & 276 & 35 & \textsuperscript{tz}\,$1.15\times$ & $1.05\times$ \\
alu-v4\_36 & \textsuperscript{G} & 118 & 36 & \textsuperscript{gq}\,$1.44\times$ & $1.11\times$ & mod5\_4\_\_o3 & \textsuperscript{F} & 40 & 10 & \textsuperscript{qu}\,$1.33\times$ & $1.05\times$ \\
ham3\_102 & \textsuperscript{GQ} & 20 & 4 & \textsuperscript{qu}\,$1.25\times$ & $1.11\times$ & 4gt12-v0\_87 & \textsuperscript{G} & 247 & 39 & \textsuperscript{tz}\,$1.19\times$ & $1.05\times$ \\
ex2\_227 & \textsuperscript{G} & 646 & 130 & \textsuperscript{tz}\,$1.25\times$ & $1.11\times$ & 4mod7-v1\_96 & \textsuperscript{G} & 170 & 33 & \textsuperscript{qu}\,$1.24\times$ & $1.05\times$ \\
4gt13-v1\_93 & \textsuperscript{G} & 74 & 12 & \textsuperscript{tz}\,$1.19\times$ & $1.10\times$ & 4gt10-v1\_81 & \textsuperscript{G} & 151 & 17 & \textsuperscript{pr}\,$1.13\times$ & $1.05\times$ \\
rd53\_138 & \textsuperscript{G} & 132 & 21 & \textsuperscript{qu}\,$1.19\times$ & $1.10\times$ & 4gt13\_92 & \textsuperscript{GQ} & 66 & 12 & \textsuperscript{qu}\,$1.22\times$ & $1.05\times$ \\
alu-v2\_32 & \textsuperscript{G} & 166 & 27 & \textsuperscript{qu}\,$1.19\times$ & $1.10\times$ & toffoli\_n3 & \textsuperscript{B} & 24 & 1 & \textsuperscript{vq}\,$1.04\times$ & $\boldsymbol{1.04\times}$ \\
cm82a\_208 & \textsuperscript{G} & 671 & 183 & \textsuperscript{pr}\,$1.38\times$ & $1.10\times$ & 4gt13\_91 & \textsuperscript{G} & 103 & 25 & \textsuperscript{qu}\,$1.32\times$ & $1.04\times$ \\
mod5d2\_64 & \textsuperscript{G} & 56 & 10 & \textsuperscript{po}\,$1.22\times$ & $1.10\times$ & ham7\_104 & \textsuperscript{G} & 320 & 42 & \textsuperscript{tz}\,$1.15\times$ & $1.04\times$ \\
one-two-three-v0\_98 & \textsuperscript{G} & 146 & 19 & \textsuperscript{qu}\,$1.15\times$ & $1.10\times$ & 4gt13\_90 & \textsuperscript{G} & 107 & 24 & \textsuperscript{qu}\,$1.29\times$ & $1.04\times$ \\
4mod5-bdd\_287 & \textsuperscript{G} & 79 & 14 & \textsuperscript{qu}\,$1.22\times$ & $1.10\times$ & barenco\_tof\_4 & \textsuperscript{FGQ} & 114 & 28 & \textsuperscript{tz}\,$1.33\times$ & $1.04\times$ \\
majority\_239 & \textsuperscript{G} & 621 & 121 & \textsuperscript{tz}\,$1.24\times$ & $1.10\times$ & alu-bdd\_288 & \textsuperscript{G} & 87 & 26 & \textsuperscript{qu}\,$1.43\times$ & $1.04\times$ \\
sat\_n7 & \textsuperscript{B} & 243 & 76 & \textsuperscript{pb}\,$1.46\times$ & $1.09\times$ & decod24-v1\_41 & \textsuperscript{G} & 91 & 14 & \textsuperscript{qu}\,$1.18\times$ & $1.03\times$ \\
decod24-enable\_126 & \textsuperscript{G} & 338 & 51 & \textsuperscript{tz}\,$1.18\times$ & $1.09\times$ & decod24-v3\_45 & \textsuperscript{G} & 165 & 24 & \textsuperscript{qu}\,$1.17\times$ & $1.03\times$ \\
4gt4-v0\_73 & \textsuperscript{G} & 395 & 74 & \textsuperscript{tz}\,$1.23\times$ & $1.09\times$ & hwb4\_49 & \textsuperscript{G} & 233 & 34 & \textsuperscript{tz}\,$1.17\times$ & $1.03\times$ \\
C17\_204 & \textsuperscript{G} & 470 & 88 & \textsuperscript{tz}\,$1.23\times$ & $1.09\times$ & 4mod7-v0\_94 & \textsuperscript{G} & 162 & 29 & \textsuperscript{qu}\,$1.22\times$ & $1.03\times$ \\
sf\_276 & \textsuperscript{G} & 808 & 162 & \textsuperscript{tz}\,$1.25\times$ & $1.09\times$ & 4gt5\_77 & \textsuperscript{G} & 134 & 20 & \textsuperscript{tz}\,$1.18\times$ & $1.02\times$ \\
4gt4-v0\_79 & \textsuperscript{G} & 231 & 46 & \textsuperscript{pr}\,$1.25\times$ & $1.08\times$ & miller\_11 & \textsuperscript{GQ} & 50 & 6 & \textsuperscript{qu}\,$1.14\times$ & $1.02\times$ \\
rd53\_135 & \textsuperscript{G} & 296 & 44 & \textsuperscript{tz}\,$1.17\times$ & $1.08\times$ & aj-e11\_165 & \textsuperscript{G} & 154 & 19 & \textsuperscript{qu}\,$1.14\times$ & $1.02\times$ \\
4gt12-v0\_88 & \textsuperscript{G} & 194 & 26 & \textsuperscript{tz}\,$1.15\times$ & $1.08\times$ & hlf\_5 & \textsuperscript{G} & 170 & 151 & \textsuperscript{pb}\,$8.95\times$ & $1.02\times$ \\
4gt4-v0\_80 & \textsuperscript{G} & 182 & 22 & \textsuperscript{tz}\,$1.14\times$ & $1.08\times$ & hwb6 & \textsuperscript{FG} & 283 & 33 & \textsuperscript{pb}\,$1.13\times$ & $1.01\times$ \\
alu-v0\_27 & \textsuperscript{GQ} & 39 & 4 & \textsuperscript{qu}\,$1.11\times$ & $1.08\times$ & one-two-three-v2\_100 & \textsuperscript{G} & 72 & 5 & \textsuperscript{qu}\,$1.07\times$ & $1.01\times$ \\
hwb6\_\_o3 & \textsuperscript{F} & 641 & 401 & \textsuperscript{qi}\,$2.67\times$ & $1.08\times$ & qft\_4 & \textsuperscript{F} & 179 & 6 & \textsuperscript{tz}\,$1.03\times$ & $1.01\times$ \\
alu-v1\_28 & \textsuperscript{GQ} & 40 & 4 & \textsuperscript{qu}\,$1.11\times$ & $1.08\times$ & cat\_state\_n4 & \textsuperscript{B} & 4 & 0 & \textsuperscript{---}\,$1.00\times$ & $\boldsymbol{1.00\times}$ \\
alu-v3\_35 & \textsuperscript{GQ} & 40 & 4 & \textsuperscript{qu}\,$1.11\times$ & $1.08\times$ & ex-1\_166 & \textsuperscript{GQ} & 22 & 0 & \textsuperscript{---}\,$1.00\times$ & $\boldsymbol{1.00\times}$ \\
alu-v4\_37 & \textsuperscript{GQ} & 40 & 4 & \textsuperscript{qu}\,$1.11\times$ & $1.08\times$ & ex1\_226 & \textsuperscript{G} & 13 & 0 & \textsuperscript{---}\,$1.00\times$ & $\boldsymbol{1.00\times}$ \\
alu-v2\_30 & \textsuperscript{G} & 510 & 101 & \textsuperscript{gq}\,$1.25\times$ & $1.08\times$ & graycode6\_47 & \textsuperscript{G} & 5 & 0 & \textsuperscript{---}\,$1.00\times$ & $\boldsymbol{1.00\times}$ \\
4gt11\_82 & \textsuperscript{Q} & 27 & 3 & \textsuperscript{qu}\,$1.12\times$ & $1.08\times$ & qrng\_n4 & \textsuperscript{B} & 4 & 0 & \textsuperscript{---}\,$1.00\times$ & $\boldsymbol{1.00\times}$ \\
ex3\_229 & \textsuperscript{G} & 412 & 76 & \textsuperscript{tz}\,$1.23\times$ & $1.08\times$ & teleportation\_n3 & \textsuperscript{B} & 8 & 0 & \textsuperscript{---}\,$1.00\times$ & $\boldsymbol{1.00\times}$ \\
4gt12-v1\_89 & \textsuperscript{G} & 234 & 44 & \textsuperscript{tz}\,$1.23\times$ & $1.08\times$ & tof\_3\_\_o3 & \textsuperscript{F} & 35 & 0 & \textsuperscript{---}\,$1.00\times$ & $\boldsymbol{1.00\times}$ \\
4gt4-v0\_78 & \textsuperscript{G} & 235 & 38 & \textsuperscript{tz}\,$1.19\times$ & $1.08\times$ & tof\_4\_\_o3 & \textsuperscript{F} & 55 & 0 & \textsuperscript{---}\,$1.00\times$ & $\boldsymbol{1.00\times}$ \\
alu-v1\_29 & \textsuperscript{GQ} & 43 & 2 & \textsuperscript{qi}\,$1.05\times$ & $\boldsymbol{1.07\times}$ &  &  &  &  &  &  \\
\bottomrule
\end{tabular}
\end{table*}

\clearpage

\begin{figure*}[p]
\centering
\includegraphics[max width=\linewidth]{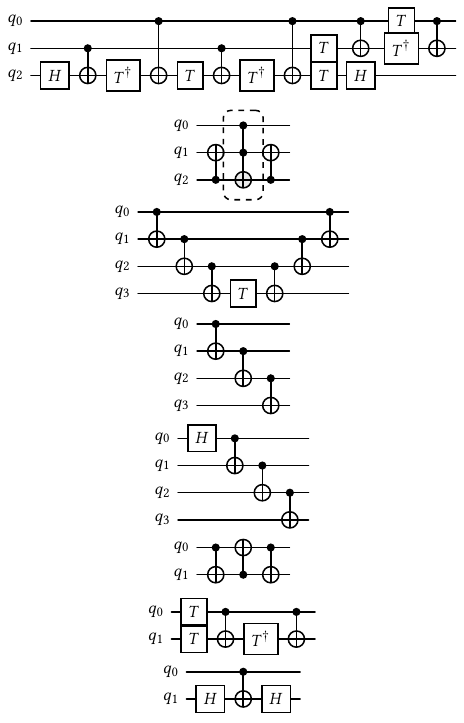}
\caption{\textbf{Motif library} (8 entries). Building blocks the circuit sampler draws from, all expressed in the Clifford+$T$ gate set $\{H,S,S^\dagger,T,T^\dagger,\textsc{cx}\}$; no parameterized rotations occur. The dashed Toffoli in the second circuit is a macro for the 15-gate circuit directly above it. The third through fifth circuits form the \emph{ladder family}: their width $k$ is sampled uniformly from $[k_{\min},Q]$ and is shown here at $k=4$.}
\Description{Eight Clifford+T circuit motifs, shown without side labels. They are, from top to bottom, a Toffoli decomposition, a Fredkin construction whose dashed box abbreviates the Toffoli above it, a phase gadget, a parity ladder, a GHZ ladder, a swap, a controlled S, and a controlled Z.}
\label{fig:motifs}
\end{figure*}
\begin{figure*}[p]
\centering
\includegraphics[max width=\linewidth]{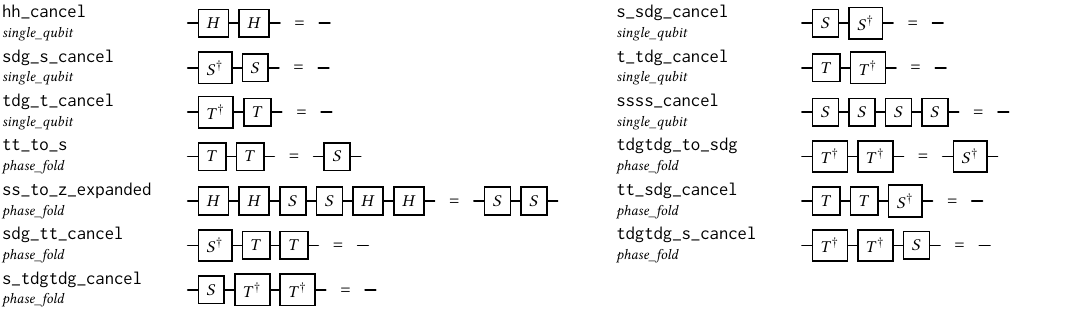}
\caption{\textbf{Single-qubit rewrite rules} (13 built-in). Family \textsf{single\_qubit} (self-inverse cancellation, 6 rules) and family \textsf{phase\_fold} (7 rules). A right-hand side drawn as a bare wire is the identity: the matched pattern is deleted.}
\Description{A table of thirteen single-qubit rewrite rules. Each entry gives the rule name and family, then the longer side of the rule as a one-qubit circuit diagram, an equals sign, and the shorter side. A right-hand side drawn as a bare wire denotes the identity.}
\label{fig:rules-1q}
\end{figure*}
\begin{figure*}[p]
\centering
\includegraphics[max width=\linewidth]{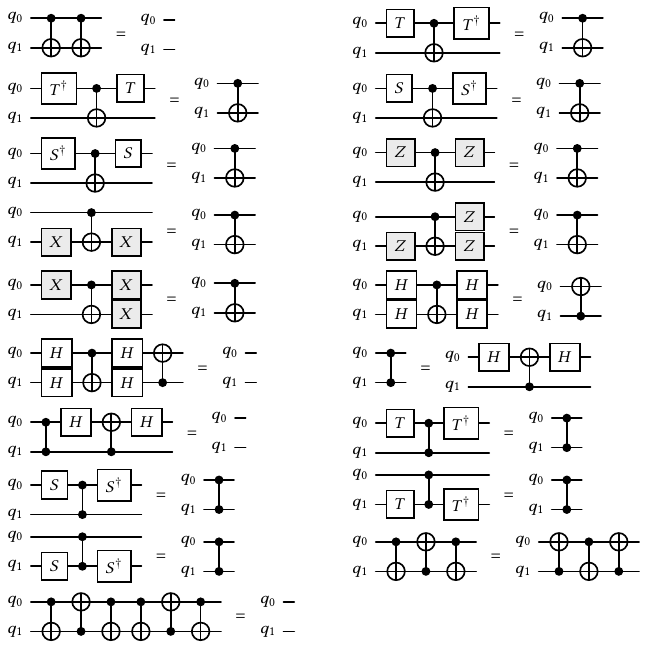}
\caption{\textbf{Two-qubit rewrite rules} (19 built-in). Composite gates use $Z=SS$, $X=HSSH$, and $CZ=H_1\,\textsc{cx}_{0,1}\,H_1$; the shaded $X$/$Z$ boxes and the $CZ$ symbol are drawn compactly. Three facts generate the whole group: diagonal gates commute through the \textsc{cx} control; \textsc{cx} conjugates $Z$ on the target to $Z_cZ_t$ and $X$ on the control to $X_cX_t$; and $H\otimes H$ conjugation reverses the \textsc{cx} direction.}
\Description{Nineteen two-qubit rewrite rules shown without side labels. Each entry shows the longer side of the rule as a two-qubit circuit diagram, an equals sign, and the shorter side. Composite X, Z and CZ gates are drawn as single shaded boxes.}
\label{fig:rules-2q}
\end{figure*}
\begin{figure*}[p]
\centering
\includegraphics[max width=\linewidth]{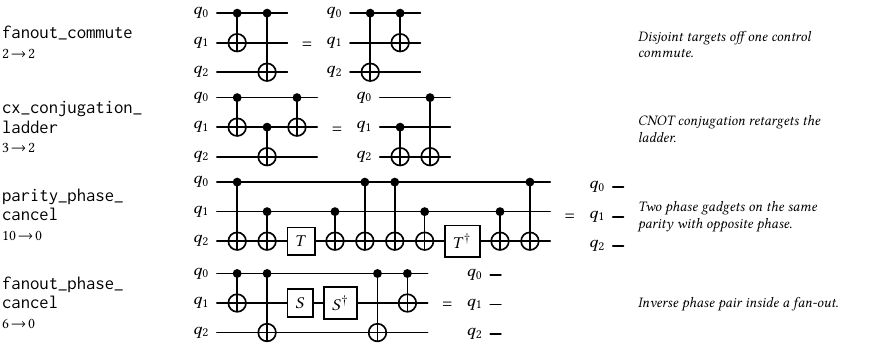}
\caption{\textbf{Three-qubit rewrite rules} (4 built-in).}
\Description{A table of four three-qubit rewrite rules, each shown as the longer side of the rule as a three-qubit circuit diagram, an equals sign, the shorter side, and a short note.}
\label{fig:rules-3q}
\end{figure*}
\begin{figure*}[p]
\centering
\includegraphics[max width=\linewidth]{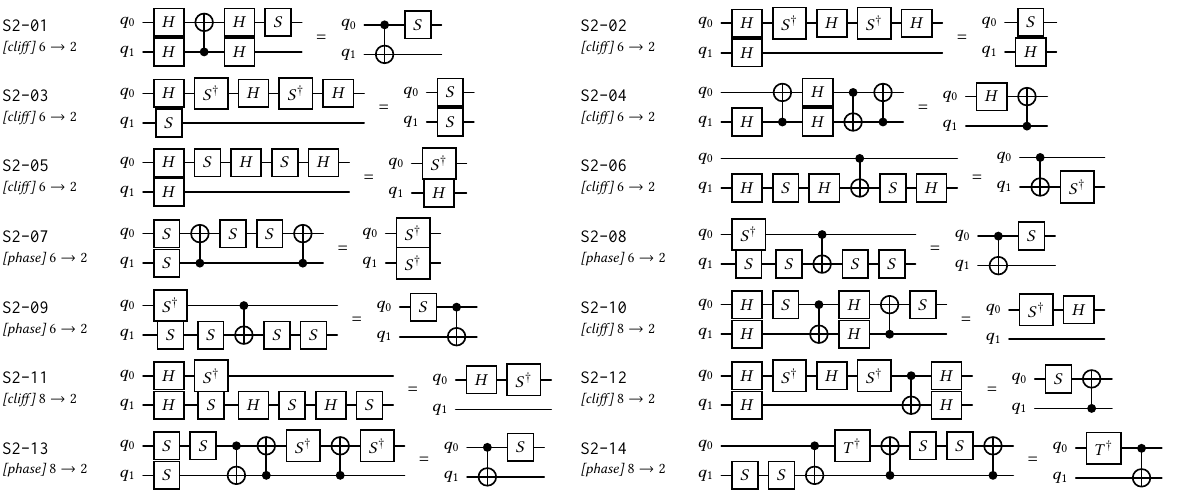}
\caption{\textbf{SAT-mined rewrite rules with two-qubit support} (14 of 50). Ruleset \texttt{q3\_t2\_restricted\_smax8} (\texttt{r50\_q3t2\_s8.json}): gate set Clifford+$T$, 3 qubits, support 2--3, source length 3--8 drawn from $\mathrm{Beta}(5,2)$, target length exactly 2, no identity targets, all-to-all connectivity, no ancillas, seed 7. Rules tagged \textsf{[cliff]} come from the \texttt{clifford\_z3} backend and hold up to a global phase (here always $e^{\pm i\pi/4}$); rules tagged \textsf{[phase]} come from \texttt{phase\_z3} and hold exactly.}
\Description{A table of fourteen SAT-mined rewrite rules acting on two qubits. Each entry gives the rule identifier, a backend tag of either cliff or phase, and the gate counts of the two sides, then the longer side as a circuit diagram, an equals sign, and the two-gate shorter side.}
\label{fig:sat-2q}
\end{figure*}
\begin{figure*}[p]
\centering
\includegraphics[max width=\linewidth]{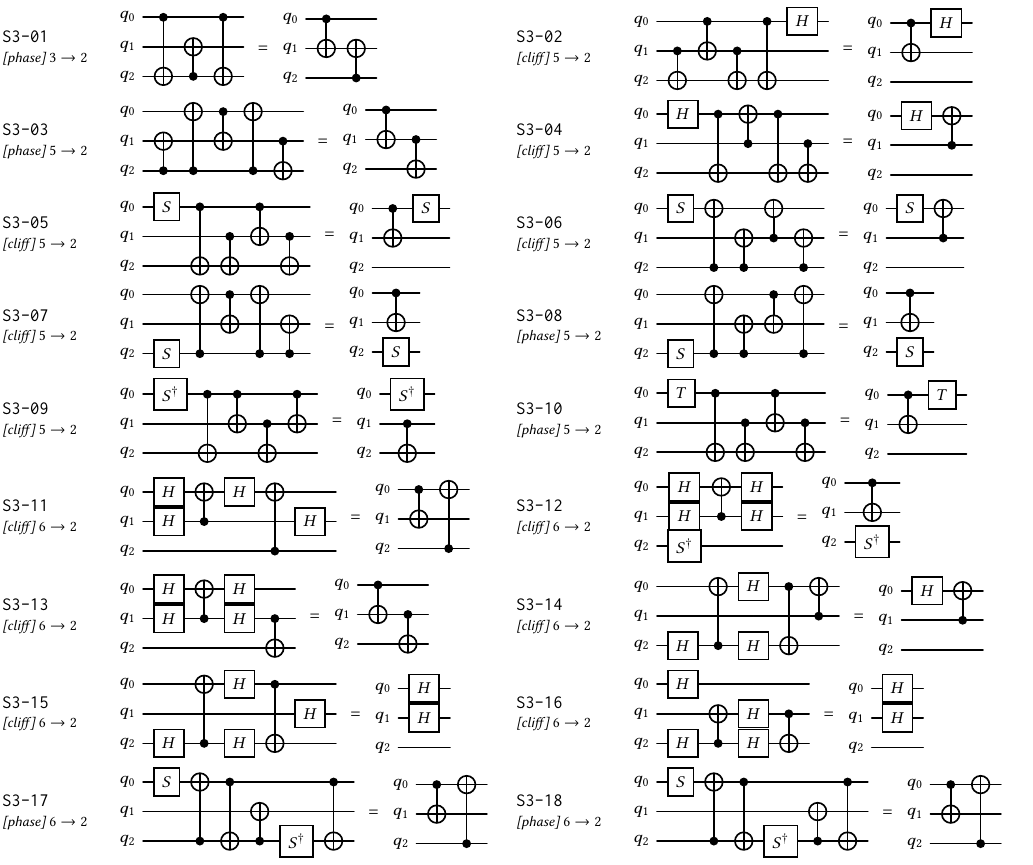}
\caption{\textbf{SAT-mined rewrite rules with three-qubit support} (36 of 50). Mining configuration and backend tags as in Fig.~\ref{fig:sat-2q}. Objective is \texttt{total\_primitive\_gate\_count}; every right-hand side has exactly two gates, so each rule removes between 1 and 6 gates. \emph{(part 1 of 2)}}
\Description{The first of two tables of SAT-mined rewrite rules acting on three qubits. Each entry gives the rule identifier, a backend tag, and the gate counts of the two sides, then the longer side as a three-qubit circuit diagram, an equals sign, and the two-gate shorter side.}
\label{fig:sat-3q}
\end{figure*}
\begin{figure*}[p]
\centering
\includegraphics[max width=\linewidth]{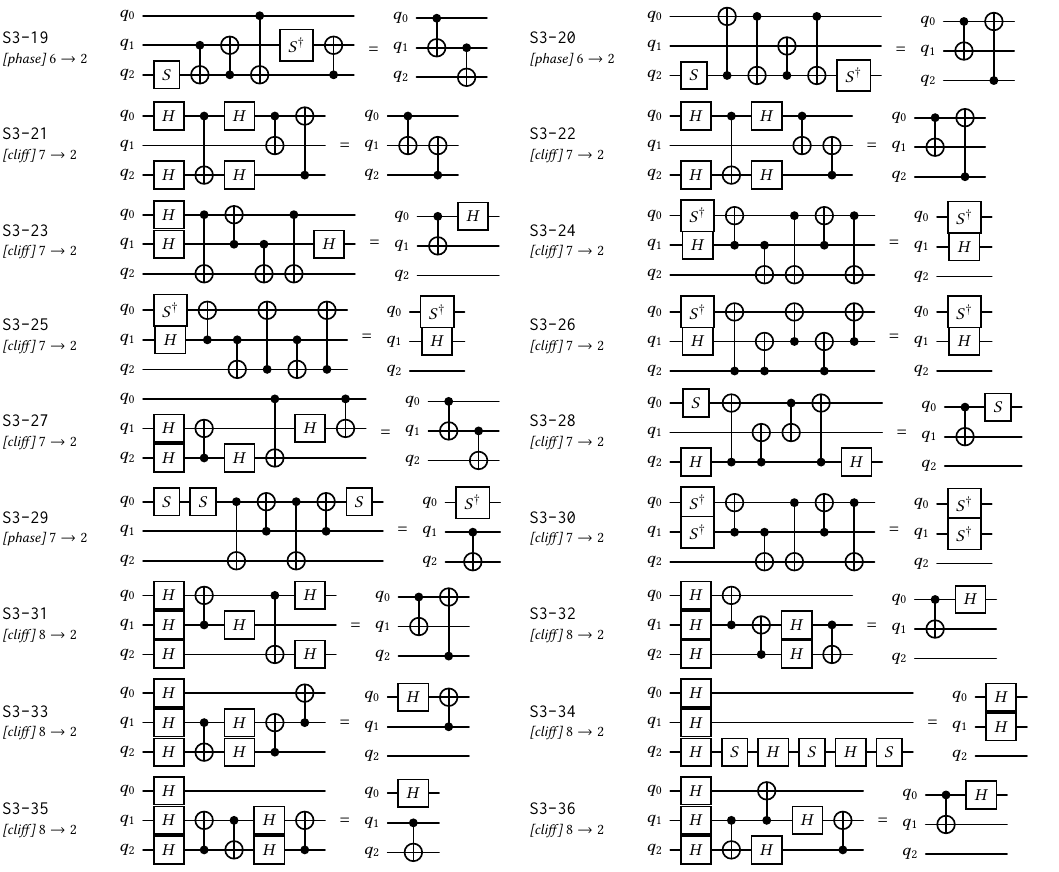}
\caption{\textbf{SAT-mined rewrite rules with three-qubit support} (36 of 50) \emph{(part 2 of 2, continued)}}
\Description{The second of two tables of SAT-mined rewrite rules acting on three qubits, laid out as the first.}
\label{fig:sat-3q-2}
\end{figure*}

\end{document}